\documentclass[11pt,dvipsnames]{extarticle}

\usepackage{amsfonts}
\usepackage{amsmath}
\usepackage{amssymb}
\usepackage{amscd}
\usepackage{amsthm}
\usepackage{mathrsfs}
\usepackage{graphicx}
\usepackage{wasysym}
\usepackage{xcolor}
\usepackage{tikz,tikz-cd}
\usepackage{geometry}
\usepackage{physics}
\usepackage[12hr]{datetime}
\usepackage{textcomp}
\usepackage{hyperref}
\usepackage[utf8]{inputenc}
\usepackage[T1]{fontenc}
\usepackage[color,notref,notcite]{showkeys}
\usepackage[nameinlink, noabbrev, capitalize]{cleveref}
\usepackage{appendix}
\usepackage{mathtools}
\usepackage{xfrac}
\usepackage{nicefrac}
\usepackage{url}
\usepackage{environ}
\usepackage{nicematrix}
\usepackage[hang]{footmisc}
\usepackage{algorithm2e}
\usepackage[normalem]{ulem}
\usepackage[shortlabels]{enumitem}
\usepackage{multicol}
\usepackage{xpatch}
\usepackage[font={small}]{caption}

\usepackage[sc,osf]{mathpazo}   
\usepackage[euler-digits,small]{eulervm}
\AtBeginDocument{\renewcommand{\hbar}{\hslash}}

\usepackage[backend=biber,style=alphabetic,backref=true,sorting=anyt,maxalphanames=3,minalphanames=2,maxbibnames=99,doi=false,isbn=false,url=false,eprint=false]{biblatex}

\renewbibmacro{in:}{}
\DeclareFieldFormat
[article,inbook,incollection,inproceedings,patent,thesis,unpublished]
{title}{#1}
\DefineBibliographyStrings{english}{
	backrefpage={},
	backrefpages={}
}

\DeclareFieldFormat{backrefparens}{\addperiod#1}
\xpatchbibmacro{pageref}{parens}{backrefparens}{}{}

\hypersetup{
	colorlinks=true, urlcolor=NavyBlue, linkcolor=Mahogany, citecolor=ForestGreen, pdfborder={0,0,0},
}

\let\oldFootnote\footnote
\newcommand\nextToken\relax

\renewcommand\footnote[1]{%
	\oldFootnote{#1}\futurelet\nextToken\isFootnote}

\newcommand\isFootnote{%
	\ifx\footnote\nextToken\textsuperscript{,}\fi}

\newcommand{\mb}{\mathbb}

\newcommand{\mc}{\mathcal}

\newcommand{\tbf}{\textbf}

\newcommand{\mbf}{\mathbf}
\newcommand{\msf}{\mathsf}
\newcommand{\mtt}{\mathtt}
\newcommand{\mrm}{\mathrm}
\newcommand{\n}{\enspace}
\newcommand{\tx}{\text}

\newcommand{\wt}{\widetilde}

\newcommand{\ord}{\tx{\normalfont ord}}

\newcommand{\iref}[2]{(\hyperref[#2]{\ref*{#1}.\ref*{#2}})}

\newcommand{\diag}{\tx{\normalfont diag}}

\newcommand{\coeff}{\tx{\normalfont coeff}}
\newcommand{\email}[1]{\href{mailto:#1}{\textcolor{NavyBlue}{\texttt{#1}}}}

\newcommand{\FRS}{\mathrm{FRS}}
\newcommand{\FRM}{\tx{\normalfont FRM}}

\newcommand{\Mult}{\mathrm{Mult}}
\newcommand{\Qmult}{\Q\mhyphen\tx{\normalfont Mult}}

\newcommand{\Q}{\mtt{Q}}

\newcommand{\charac}{\tx{\normalfont char}}

\DeclareRobustCommand{\sqbinom}{\genfrac[]{0pt}{}}
\newcommand{\Qbinom}[2]{\sqbinom{#1}{#2}_\Q}

\newcommand{\DQ}{\msf{D}_\Q}
\newcommand{\DQi}[1]{\msf{D}_{\Q,#1}}

\newcommand\nodot[1]{}

\mathchardef\mhyphen="2D

\usepackage{thmtools}
\usepackage{thm-restate}

\theoremstyle{theorem}
\newtheorem{theorem}{Theorem}[section]

\newtheorem*{claim*}{Claim}

\newtheorem{proposition}[theorem]{Proposition}

\newtheorem{lemma}[theorem]{Lemma}
\newtheorem{corollary}[theorem]{Corollary}

\newtheoremstyle{TheoremNum}
{\topsep}{\topsep}              
{\itshape}                      
{}                              
{\bfseries}                     
{.}                             
{ }                             
{\thmname{#1}\thmnote{ \bfseries #3}}
\theoremstyle{TheoremNum}

\theoremstyle{definition}

\newtheorem{remark}[theorem]{Remark}

\makeatletter
\NewEnviron{restatethm}[2]{%
	\protected@write\@auxout{}{%
		\string\@restatetheorem{#1}{\detokenize\expandafter{\BODY}}%
	}%
	\begin{theorem}[#2]\label{#1}\BODY\end{theorem}%
}
\newcommand{\@restatetheorem}[2]{%
	\expandafter\gdef\cNsame restatethm@#1\endcNsame{#2}%
}

\makeatother

\makeatletter
\newcommand*{\transpose}{%
	{\mathpalette\@transpose{}}%
}
\newcommand*{\@transpose}[2]{%
	\raisebox{\depth}{$\m@th#1\intercal$}%
}
\makeatother



\title{Polynomials, Divided Differences, and Codes~\footnote{A preliminary version of this work is due to appear at ITCS 2025.  In the preliminary version, we present our algorithmic result for the list recovery problem.  Since it is an obvious extension (\`a la~\cite{guruswami-wang-2013-FRS}) of our algorithmic result for list decoding, we only present the list decoding algorithm in this version to keep things simple.}}
\author{S. Venkitesh~\thanks{Blavatnik School of of Computer Science, Tel Aviv University, Tel Aviv, Israel.  Supported by Len Blavatnik and the Blavatnik Family Foundation.  A part of this work was done while the author was a postdoc at the Department of Computer Science, University of Haifa, supported in part by BSF grant 2021683, ISF grant 735/20, and by the European Union (ERC, ECCC, 101076663), as well as while the author was participating in the program on Error-Correcting Codes and Computation (2024) at the Simons Institute for the Theory of Computing, UC Berkeley.\newline\url{https://sites.google.com/view/venkitesh}.}\\Tel Aviv University\\\email{venkitesh.mail@gmail.com}}
\date{}

\begin{document}
	
	\maketitle
	
	\begin{abstract}
		Multivariate multiplicity codes (Kopparty, Saraf, and Yekhanin, J. ACM 2014) are linear codes where the codewords are described by evaluations of multivariate polynomials (with a degree bound) and their derivatives up to a fixed order, on a suitably chosen affine point set.  While good list decoding algorithms for multivariate multiplicity codes were known in some special cases of point sets by a reduction to univariate multiplicity codes, a general list decoding algorithm up to the distance of the code when the point set is an arbitrary finite grid, was obtained only recently (Bhandari et al., IEEE TIT 2023).  This required the characteristic of the field to be zero or larger than the degree bound, and this requirement is somewhat necessary, since list decoding this code up to distance with small output list size is not possible when the characteristic is significantly smaller than the degree.
		
		In this work, we present an alternate construction, based on divided differences, that closely resembles the classical multiplicity codes but is \emph{insensitive to the field characteristic}.  We obtain an efficient algorithm that list decodes this code up to distance, for arbitrary finite grids and over all finite fields.  Notably, our construction can be interpreted as a \emph{folded Reed-Muller code}, which might be of independent interest.  The upshot of our result is that a good \emph{Taylor-like expansion} can be expressed in terms of a good \emph{derivative-like operator} (a divided difference), and this implies that the corresponding code admits good algorithmic list decoding.
%
%
	\end{abstract}
	
	\newpage
	\tableofcontents
	
	\newpage
	\section{Introduction and overview}\label{sec:intro}
	
	Codes based on polynomial evaluations have found widespread applications, both implicitly and explicitly, in theoretical computer science, combinatorics, and allied areas.  The prototypical univariate polynomial code is the \emph{Reed-Solomon (RS) code}~\cite{reed,reed-solomon-1960}, wherein the codewords are evaluations of bounded degree univariate polynomials on a set of points, and the prototypical multivariate polynomial code is the analogously defined \emph{Reed-Muller (RM) code}~\cite{muller-1954-RM-codes,reed,kasami-lin-peterson-1968-RM-codes,weldon-1968-RM-codes,delsarte-goethals-macwilliams-1970-RM-codes}.~\footnote{The RM code was considered over the binary field by~\cite{muller-1954-RM-codes,reed}, and over larger fields by~\cite{kasami-lin-peterson-1968-RM-codes,weldon-1968-RM-codes,delsarte-goethals-macwilliams-1970-RM-codes}.}  Despite several decades of research, the decodability properties of these codes are far from being fully understood.
	
	\subsection{The list decoding problem}
	
	In this work, we are specifically interested in the \emph{list decoding problem}~\cite{elias-1957-list-decoding,wozencraft-1958-list-decoding}, where the aim is to output a list of close codewords for any given received word efficiently.  The main parameter of interest is the \emph{list decoding radius}, which represents the fraction of disagreements that we can decode from, and we would like it to be as large as possible.
	
	\subsection*{Unique decoding of polynomial codes}
	
	To begin with, consider the easier setting where we fix the list decoding radius to be \emph{half the minimum distance} of the code, which is precisely the setting of the \emph{unique decoding problem} -- in this case, there could only be at most one codeword this close to any given received word, and the aim is to efficiently find it if it exists.  This problem has been solved completely for RS codes~\cite{peterson,massey,welch-1983}, for univariate multiplicity codes~\cite{nielsen-2001-multiplicity-unique-decoding}, for RM codes in a special case~\cite{kasami-lin-peterson-1968-polynomial-codes,babai-fortnow-levin-szegedy-1991-computation-polylog-time}, for RM codes in general~\cite{kim-kopparty-2017}, and recently for multivariate multiplicity codes~\cite{bhandari-harsha-kumar-shankar-2023-multiplicity-SZ}.  What is noteworthy is that all of these algorithms are \emph{insensitive to the field characteristic}.  We now move on to consider larger list decoding radius.
	
	\subsection*{List decoding of univariate polynomial codes}
	
	In the last decade of the 20th century, two classic works~\cite{sudan-1997-RS,guruswami-sudan-1998-RS-list-decoding} showed that all RS codes can be list decoded up to the \emph{Johnson bound} (which is relative radius \(1-\sqrt{R}\), where \(R\) is the rate of the code), and also laid an algebraic algorithmic framework that has since been exploited over and over again for several other polynomial codes.  Further, we do not know of explicit RS codes that can be list decoded beyond the Johnson bound, but we do know of explicit RS codes that are not list decodable significantly beyond the Johnson bound~\cite{ben-sasson-kopparty-radhakrishnan-2010-subspace-polynomials}.  Due to recent breakthroughs~\cite{brakensiek-gopi-makam-2023-random-RS,guo-zhang-2023-random-RS,alrabiah-guruswami-li-2023-random-RS}, we now know that random RS codes are list decodable \emph{up to (information theoretic) capacity} (that is, relative radius \(1-R-\epsilon\) for arbitrarily small but fixed \(\epsilon>0\), which represents getting approximately close to the information theoretic limit of list decoding), but we neither know of an explicit construction nor know of an algorithm to achieve this.
	
	It is in this context of list decoding up to information theoretic capacity that a \emph{folding} trick has been enlightening.  The works of~\cite{guruswami-rudra-2008-FRS,guruswami-wang-2013-FRS} showed that \emph{folded Reed-Solomon (FRS) codes} (where the bits of an RS codeword are bunched into blocks of large constant size), and \emph{univariate multiplicity codes}~\cite{kopparty-2015-multiplicity-code} (wherein the codewords are evaluations of polynomials and their derivatives up to a large constant multiplicity) can be algorithmically list decoded up to capacity!  Further improvements to the output list size have been achieved in~\cite{kopparty-ron-zewi-saraf-wootters-2023-list-decoding,tamo-2023-tighter-list-size,srivastava-2024-FRS,chen-zhang-2024-FRS-list-size} in different settings.
	
	\subsection*{List decoding of multivariate polynomial codes}
	
	Moving on to multivariate polynomial codes, while the list decoding algorithms of~\cite{guruswami-sudan-1998-RS-list-decoding} generalize easily to RM codes, the decoding radius worsens as a function of the number of variables, and attaining Johnson bound is no longer possible with this algorithm.  While we know of a Johnson bound attaining algorithm~\cite{pellikaan-wu-2004-list-decoding-RM-codes} for RM codes when the finite grid is the full vector space \(\mb{F}_q^m\) over the base field \(\mb{F}_q\), designing a more general algorithm for arbitrary finite grids is still open!
	
	As it turns out, the folding trick helps in the multivariate setting too.  The multivariate multiplicity codes~\cite{kopparty-saraf-yekhanin-2014-multiplicity-codes}, which are the obviously defined multivariate analogue of univariate multiplicity codes, were shown to be list decodable \emph{up to distance} (that is, up to radius \(\delta-\epsilon\), where \(\delta\) is the minimum distance of the code) by~\cite{kopparty-2015-multiplicity-code} when the set of evaluation points is the full vector space \(\mb{F}_q^m\).  However, just as in the case of RM codes, an analogous result for arbitrary finite grids evaded us for quite some time.  Very recently,~\cite{bhandari-harsha-kumar-sudan-2024-multiplicity-code} finally bolstered the algebraic algorithmic framework with some simple additional ingredients and managed to design an algorithm for list decoding arbitrary multivariate multiplicity codes up to distance.
	
	\subsection*{Sensitivity of list decodability to the field characteristic}
	
	Finally, yet another crucial component that could influence the performance of the code (specifically a polynomial code) is the most fundamental underlying algebraic object -- the base field over which the polynomials are considered, which is always a finite field in our discussion.
	
	Unlike the unique decoding algorithms that we mentioned earlier, we do observe sensitivity to the field characteristic in the setting of the list decoding problem.  Let us fix the notation that \(\mb{F}_q\) denotes the base field, and \(p\) denotes the characteristic of \(\mb{F}_q\).  On one hand, we know that when \(q\) is a large power of \(p\), there are RS codes that cannot achieve list decoding capacity~\cite{ben-sasson-kopparty-radhakrishnan-2010-subspace-polynomials} with small output list size.  On the other hand, upon large constant folding, the univariate multiplicity codes achieve capacity when \(p\) is larger than the degree~\cite{guruswami-wang-2013-FRS,kopparty-2015-multiplicity-code}, or smaller but linear in the degree parameter~\cite{kopparty-ron-zewi-saraf-wootters-2023-list-decoding}, and the FRS codes achieve list decoding capacity insensitive to how \(p\) is relative to the degree~\cite{guruswami-rudra-2008-FRS,guruswami-wang-2013-FRS}.  Thus, the list decodability of the FRS code is \emph{insensitive to the field characteristic}, whereas the list decodability of the multiplicity code is sensitive to the field characteristic.  Indeed, in the case when \(q\) is a large power of \(p\), the large output lists shown by~\cite{ben-sasson-kopparty-radhakrishnan-2010-subspace-polynomials} for RS codes can be used to construct large output lists for multiplicity codes, and so the list decodability is provably poor.
	
	Things get even more interesting in the multivariate setting, since the multivariate multiplicity codes achieve distance when \(p\) is larger than the degree parameter~\cite{pellikaan-wu-2004-list-decoding-RM-codes,bhandari-harsha-kumar-sudan-2024-multiplicity-code}, but we do not know of a field characteristic insensitive construction of a multivariate polynomial code that algorithmically achieves distance.  Furthermore, we also do not know of an appropriate multivariate analogue of the field characteristic insensitive univariate FRS code.  In this work, we answer both questions with a single code construction, along with an algorithm for list decoding up to distance.  In particular, our code subsumes the list decoding performance of the field characteristic sensitive multivariate multiplicity code~\cite{bhandari-harsha-kumar-sudan-2024-multiplicity-code}, and also provides a natural \emph{folded Reed-Muller code} construction.  The motivation for our construction comes from an extremely simple yet foundational observation within the univariate world itself -- the FRS code is also a multiplicity code!
	
	\subsection{FRS codes, univariate multiplicity codes, and divided differences}\label{sec:FRS-mult-DD}
	
	Fix a field \(\mb{F}_q\).  Consider distinct nonzero points \(a_1,\ldots,a_n\in\mb{F}_q\).  Let \(\gamma\in\mb{F}_q^\times\) be a multiplicative generator, and suppose \(s\ge1\) such that the points \(\gamma^ja_i,\,j\in[0,s-1],\,i\in[n]\) are all distinct.  For \(k\in[sn]\), the \tbf{degree-\(k\) Folded Reed-Solomon (FRS) code} is defined by
	\[
	\FRS_s(a_1,\ldots,a_n;k)=\left\{[f]_{\FRS}\coloneqq\begin{bmatrix}
		f(a_1)&\cdots&f(a_n)\\f(\gamma a_1)&\cdots&f(\gamma a_n)\\\vdots&\ddots&\vdots\\f(\gamma^{s-1}a_1)&\cdots&f(\gamma^{s-1}a_n)
	\end{bmatrix}:f(X)\in\mb{F}_q[X]\deg(f)<k\right\}\subseteq\big(\mb{F}_q^s\big)^n.
	\]
	Here, the distance function is the \tbf{Hamming distance} on alphabet \(\mb{F}_q^s\), that is, the \tbf{Hamming weight}~\footnote{We are only concerned with codes that are linear over the base field \(\mb{F}_q\), and so the Hamming distance between two codewords \(c_1,c_2\) is always equal to the Hamming weight of the codeword \(c_1-c_2\).} of a codeword \([f]_{\FRS}\) is the number of \(i\in[n]\) such that \(\begin{bmatrix}
		f(a_i)&f(\gamma a_i)&\cdots&f(\gamma^{s-1}a_i)
	\end{bmatrix}^\transpose\) is a nonzero vector in \(\mb{F}_q^s\).

	On the other hand, simply assuming \(a_1,\ldots,a_n\in\mb{F}_q\) are distinct, for \(k\in[sn]\), the \tbf{degree-\(k\) univariate multiplicity code} is defined by
	\[
	\Mult_s(a_1,\ldots,a_n;k)=\left\{[f]_{\Mult}\coloneqq\begin{bmatrix}
		f^{(0)}(a_1)&\cdots&f^{(0)}(a_n)\\f^{(1)}(a_1)&\cdots&f^{(1)}(a_n)\\\vdots&\ddots&\vdots\\f^{(s-1)}(a_1)&\cdots&f^{(s-1)}(a_n)
	\end{bmatrix}:f(X)\in\mb{F}_q[X]\deg(f)<k\right\}\subseteq\big(\mb{F}_q^s\big)^n,
	\]
	where \(f^{(j)}(X)\coloneqq\frac{\mathrm{d}^jf(X)}{\mathrm{d}X^j}\) (the \(j\)-th classical derivative) for all \(j\in[0,s-1]\).~\footnote{Typically, while working over finite fields, there will also be a \(j!\) scaling factor multiplied with the classical derivative to prevent some pathologies about ``higher multiplicity vanishing'' at a point.  With  this scaling factor, the derivative is usually called the \tbf{Hasse derivative}.  We do not consider this, and implicitly assume that the characteristic is large enough so that these pathologies do not arise.  Indeed, this is precisely the setting where the multiplicity codes have good list decodability.}  Once again, the Hamming weight of a codeword \([f]_{\Mult}\) is the number of \(i\in[n]\) such that \(\begin{bmatrix}
		f^{(0)}(a_i)&f^{(1)}(a_i)&\cdots&f^{(s-1)}(a_i)
	\end{bmatrix}^\transpose\) is a nonzero vector in \(\mb{F}_q^s\).

	Let us begin by looking at an algebraic feature that is intrinsic to multiplicity codes.  Consider any \(f(X)\in\mb{F}_q[X]\) with \(\deg(f)<k\).  For any \(a\in\mb{F}_q\), we get the standard Taylor expansion
	\[
	f(X)=f^{(0)}(a)+f^{(1)}(Y)(X-a)+f^{(2)}(a)\frac{(X-a)^2}{2!}+\cdots+f^{(k-1)}(a)\frac{(X-a)^{k-1}}{(k-1)!},
	\]
	provided \(p\coloneqq\charac(\mb{F}_q)\ge k\).  Is there an analogous Taylor-like expansion in the context of the FRS code?  As it turns out, there is such an expansion, and it is even simpler -- it is the expansion in terms of \emph{Newton forward differences}.  Precisely, we have
	\[
	f(X)=f^{[0]}(a)+f^{[1]}(a)(X-a)+f^{[2]}(a)(X-a)(X-\gamma a)+\cdots+f^{[k-1]}(a)(X-a)\cdots(X-\gamma^{k-1}a),
	\]
	where \(f^{[0]}(a)=f(a)\), and
	\[
	f^{[1]}(a)=\frac{f^{[0]}(\gamma a)-f^{[0]}(a)}{(\gamma-1)a},\n f^{[2]}(a)=\frac{f^{[1]}(\gamma a)-f^{[1]}(a)}{(\gamma^2-1)a},\n\ldots,\n f^{[k-1]}(a)=\frac{f^{[k-2]}(\gamma a)-f^{[k-2]}(a)}{(\gamma^{k-1}-1)a}.
	\]
	Notably, since \(\gamma\in\mb{F}_q^\times\) is a multiplicative generator, we always have \(\ord(\gamma)=q-1\ge k\), and therefore, this expansion is valid \emph{unconditional on \(p\)}.
	
	In the first decade of the 20th century, Jackson~\cite{jackson-1909-q-taylor-theorem,jackson-1909-q-taylor-series,jackson-1909-q-functions} (also see~\cite{jackson-1910-q-difference-equations,jackson-1910-q-integrals}) considered the following specific \emph{divided difference} operator
	\[
	\msf{D}_\gamma f(X)\coloneqq\frac{f(\gamma X)-f(X)}{(\gamma-1)X},\n\tx{and}\n\msf{D}_\gamma^{t+1}\coloneqq\msf{D}_\gamma\circ\msf{D}_\gamma^t\n\tx{for all }t\ge1.
	\]
	In terms of this operator, the above Taylor-like expansion would become
	\[
	f(X)=f(a)+\msf{D}_\gamma(a)(X-a)+\msf{D}_\gamma^{2}(a)\frac{(X-a)(X-\gamma a)}{[2]_\gamma!}+\cdots+\msf{D}_\gamma^{k-1}(a)\frac{(X-a)\cdots(X-\gamma^{k-1}a)}{[k-1]_\gamma!},
	\]
	where \([t]_\gamma\coloneqq\frac{\gamma^t-1}{\gamma-1}=1+\gamma+\gamma^2+\cdots+\gamma^{t-1}\), and \([t]_\gamma!\coloneqq[t]_\gamma[t-1]_\gamma\cdots[1]_\gamma\).  Returning to the FRS code \(\FRS_s(a_1,\ldots,a_n;k)\), it is then immediate by definition of \(D_\gamma\) that there are invertible lower triangular matrices \(U(a_i)\in\mb{F}_q^{s\times s},\,i\in[n]\) such that for any codeword \([f]_{\FRS}\), we have
	\[
	\begin{bmatrix}
		f(a_i)\\\msf{D}_\gamma f(a_i)\\\vdots\\\msf{D}_\gamma^{s-1}f(a_i)
	\end{bmatrix}=U(a_i)\cdot\begin{bmatrix}
	f(a_i)\\f(\gamma a_i)\\\vdots\\f(\gamma^{s-1}a_i)
	\end{bmatrix}\quad\tx{for all }i\in[n].
	\]
	In other words, after a specific change of basis, the FRS code can be interpreted as a \emph{multiplicity code} with respect to the \emph{derivative operator} \(\msf{D}_\gamma\).
	
	Note that we know two list decoding algorithms for FRS codes~\cite{guruswami-rudra-2008-FRS,guruswami-wang-2013-FRS}, with the first one being algebraic and the second one being linear algebraic.  If we are given the \emph{\(\msf{D}_\gamma\)-multiplicity encoding}, we may just change the basis to the usual FRS encoding and run one of these two FRS decoders.  However, what is really encouraging is that we can straightforwardly adapt the linear algebraic multiplicity decoder~\cite{guruswami-wang-2013-FRS} for our \(\msf{D}_\gamma\)-multiplicity encoding, and the analysis will have a strong resemblance with that in the case of the classical derivatives.  In this work, our main contribution is a construction that pushes this resemblance to the multivariate setting.
	
	\paragraph*{Assumption on the field.}  Since our intention is to present the basic ideas as clearly and quickly as possible, we will stick to a simpler setting henceforth, where we assume that the evaluation points are contained in a finite field \(\mb{F}_q\), but the polynomials are over a degree-3 field extension \(\mb{K}=\mb{F}_{q^3}\).  We can relax this setting to some extent by tinkering with the parameters involved, but we will not make this effort.
	
	\paragraph*{The \(\Q\)-derivative.}  We will now forego the use of the symbol \(\gamma\) for the multiplicative generator.  Instead we will denote \(\Q\in\mb{K}^\times\) to be a multiplicative generator, and henceforth refer to the operator \(\DQ\) as the \tbf{\(\Q\)-derivative}.  This is the more standard terminology in the literature where this operator has appeared before.  The \(\Q\)-derivative~\footnote{We use the notation `\(\Q\)' in \(\Q\)-derivative in this work instead of the more popular `\(q\)', since we use \(q\) to denote the field size.} is also called the \emph{Jackson derivative}, and is a special case of the \emph{Hahn derivative}~\cite{hahn-1949-derivative}, as well as the \emph{Newton forward difference}~\cite[Chapter 1, Section 9]{jordan-1950-finite-calculus}.  See~\cite[Chapter 26]{kac-cheung-2002-quantum-calculus} for a discussion on a slightly more general \emph{quantum differential}, as well as some symmetrized variants.  These are all instances of a broader notion of \emph{divided difference} in \emph{finite calculus}~\cite{lacroix-1819-finite-calculus,boole-1860-finite-calculus,norlund-1924-finite-calculus,steffensen-1927-interpolation,jordan-1950-finite-calculus,richardson-1954-finite-calculus}.  The \(\Q\)-derivative seems to have been used so far only over fields of characteristic zero, particularly in \emph{\(q\)-combinatorics}~\cite{exton-1983-q-hypergeometric,andrews-1986-q-series,gasper-rahman-2004-hypergeometric,roman-2005-umbral-calculus,foata-han-2021-q-series} and \emph{quantum calculus}~\cite{kac-cheung-2002-quantum-calculus,ernst-2012-q-calculus,larsson-silvestrov-2003-q-difference}.  Further, the \(\Q\)-derivative does not seem to have made any explicit appearance so far in the \emph{polynomial method} literature.  However, as we see in this work, this notion turns out to be useful over fields of positive characteristic, and indeed, as part of the polynomial method.
	
	\subsection{Multivariate \(\Q\)-derivatives, multivariate \(\Q\)-multiplicity code, and \emph{folded} RM code}\label{sec:multi-Q-FRM}
	
	Our extension of \(\Q\)-derivatives from the univariate to multivariate setting is natural, and mimicks the extension of the classical derivatives.  Assume indeterminates \(\mb{X}=(X_1,\ldots,X_m)\), and for any \(f(\mb{X})\in\mb{K}[\mb{X}]\) and \(\alpha\in\mb{N}^m\), we define the \tbf{\(\alpha\)-th \(\Q\)-derivative} as the iterated operator
	\[
	\DQ^\alpha f(\mb{X})\coloneqq\DQi{X_1}^{\alpha_1}\cdots\DQi{X_m}^{\alpha_m}f(\mb{X}),
	\]
	where \(\DQi{X_i}\) denotes the univariate \(\Q\)-derivative in the variable \(X_i\).  It is easy to see that the operator \(\DQ^\alpha\) is well-defined and invariant under the order of iterations, just as in the case of the classical partial derivatives.
	
	Now consider indeterminates \(\mb{Y}=(Y_1,\ldots,Y_m)\).  For any \(t\ge0\), denote \((X_i-Y_i)_\Q^{(t)}\coloneqq(X_i-Y_i)(X_i-\Q Y_i)\cdots(X_i-\Q^{t-1}Y_i)\), and further, for any \(\alpha\in\mb{N}^m\), denote \((\mb{X}-\mb{Y})_\Q^{(\alpha)}\coloneqq\prod_{i=1}^m(X_i-Y_i)^{(\alpha_i)}_\Q\).  Also denote \([\alpha]_\Q!\coloneqq[\alpha_1]_\Q!\cdots[\alpha_m]_\Q!\).  We will note that the two important algebraic features of unvariate \(\Q\)-derivatives that we saw in~\cref{sec:FRS-mult-DD} extend in an \emph{inductive} fashion to the multivariate setting.  For any \(f(\mb{X})\in\mb{K}[\mb{X}]\), we have the multivariate Taylor-like expansion
	\[
	f(\mb{X})=\sum_{\alpha\in\mb{N}^m}\frac{\DQ^\alpha f(\mb{Y})}{[\alpha]_\Q!}(\mb{X}-\mb{Y})_\Q^{(\alpha)}.
	\]
	Further, for any \(\alpha\in\mb{N}^m\), denote \(\Q^\alpha\mb{X}=(\Q^{\alpha_1}X_1,\ldots,\Q^{\alpha_m}X_m)\).  Then for any \(s\ge1\) and \(a\in(\mb{F}_q^\times)^m\), there exists an invertible matrix \(U(a)\in\mb{K}^{\binom{m+s-1}{s-1}\times\binom{m+s-1}{s-1}}\) such that
	\[
	\begin{bmatrix}
		\DQ^\gamma f(a)
	\end{bmatrix}_{|\gamma|<s}=U(a)\cdot\begin{bmatrix}
	f(\Q^\gamma a)
\end{bmatrix}_{|\gamma|<s}\qquad\tx{for all }f(\mb{X})\in\mb{K}[\mb{X}].
	\]
	
	Fix any \(s\ge1\), and consider any nonempty \(A\subseteq\mb{F}_q^\times\).  For any \(k\in[s|A|]\), we define the \tbf{\(m\)-variate multiplicity-\(s\) degree-\(k\) \(\Q\)-multiplicity code} by
	\[
	\Qmult_{m,s}(A;k)=\Big\{[\DQ\mid f]\coloneqq\begin{bmatrix}
		\begin{bmatrix}
			\DQ^\gamma f(a)
		\end{bmatrix}_{|\gamma|<s}
	\end{bmatrix}_{a\in A^m}:f(\mb{X})\in\mb{K}[\mb{X}],\,\deg(f)<k\Big\}\subseteq\big(\mb{K}^{\binom{m+s-1}{s-1}}\big)^{|A|^m}.
	\]
	Note that the distance function is now the Hamming distance on the alphabet \(\mb{K}^{\binom{m+s-1}{s-1}}\).  We will prove that the rate of this code is equal to \(\dfrac{\binom{m+k-1}{k-1}}{\binom{m+s-1}{s-1}|A|^m}\), and the minimum distance is at least \(1-\dfrac{k-1}{s|A|}\).  Further, by the change of basis that we just saw above, we can define a natural multivariate analogue of the univariate FRS codes.  For any \(k\in[s|A|]\), we define the \tbf{\(s\)-folded degree \(k\) folded Reed-Muller (FRM) code} by
	\[
	\FRM_{m,s}(A;k)=\Big\{[\Q\mid f]\coloneqq\begin{bmatrix}
		\begin{bmatrix}
			f(\Q^\gamma a)
		\end{bmatrix}_{|\gamma|<s}
	\end{bmatrix}_{a\in A^m}:f(\mb{X})\in\mb{K}[\mb{X}],\,\deg(f)<k\Big\}.
	\]
	So we clearly have the change of basis
	\[
	\Qmult_{m,s}(A;k)=\diag(U(a):a\in A^m)\cdot\FRM_{m,s}(A;k).
	\]
	
	\subsection{Our results}\label{sec:results}
	
	Our first result is that over \emph{any finite field} \(\mb{F}_q\) with sufficiently large \(q\), and for sufficiently large constant multiplicity \(s\), any constant rate multivariate \(\Q\)-multiplicity code \(\Qmult_{m,s}(A;k)\), can be efficiently list decoded up to distance.  
	\begin{theorem}\label{thm:multivariate-LD}
		Consider any constant \(\delta\in(0,1),\,\epsilon\in(0,\delta)\).  For any \(A\subseteq\mb{F}_q^\times\) and for the choices \(s=O(m^2/\epsilon^{2m})\) and \(k=\lceil(1-\delta)s|A|\rceil\), the code \(\Qmult^m_s(A;k)\) is efficiently list decodable up to radius \(\delta-\epsilon\) with output list contained in a \(\mb{K}\)-affine space of dimension at most \(O(m^2/\epsilon)\).
	\end{theorem}
	\noindent  Our list decoding algorithm is conceptually simpler than by~\cite{bhandari-harsha-kumar-sudan-2024-multiplicity-code} for the classical multivariate multiplicity codes.  Their algorithm in the classical setting entailed recovering the close enough polynomial messages one homogeneous component at a time, while using the classical \emph{Euler's formula for homogeneous polynomials} (that relates a homogeneous polynomial to its first order derivatives)~\footnote{Euler's formula states that for a homogeneous degree-\(d\) polynomial \(f(\mb{X})\) over a field of characteristic greater than \(d\), we have \(\sum_{i=1}^mX_i\cdot\frac{\partial f(\mb{X})}{\partial X_i}=d\cdot f(\mb{X})\).} to fully recover a homogeneous component (or a partial derivative of it) from its further first order partial derivatives.  This formula, as well as the Taylor expansion for classical derivatives employed within the recovery step, both require the field characteristic to be large.  In our \(\Q\)-setting, firstly there is no Euler's formula involving \(\Q\)-derivatives, and interestingly we don't need one!  Secondly, the Taylor expansion for \(\Q\)-derivatives is clearly field characteristic insensitive.  These are the only two places in the analysis where conceptually, the field characteristic was relevant in the classical setting, and is now irrelevant in the \(\Q\)-setting!  Further, the short and clean analysis that we obtain encourages us to believe that this construction is the correct way to obtain a folded version of the RM code.
	
	Let us give a quick outline of our algorithm.  We make use of a clean \emph{gluing} trick from~\cite{bhandari-harsha-kumar-sudan-2024-multiplicity-code}~\footnote{As mentioned in~\cite{bhandari-harsha-kumar-sudan-2024-multiplicity-code}, this gluing trick has appeared earlier in a hitting set generator construction by~\cite{guo-kumar-saptharishi-solomon-2022-derandomization}.} and then proceed to perform a simpler analysis that is extremely close to the univariate analysis of~\cite{guruswami-wang-2013-FRS}.  The informal takeaway from our result is that
	\[
	\tx{gluing trick}\n+\n\tx{univariate list decoding analysis}\n=\n\tx{multivariate list decoding analysis}.
	\]
	Suppose we are given the code \(\Qmult_{m,s}(A;k)\) as in the statement of~\cref{thm:multivariate-LD}.  Consider any received word
	\[
	w=\Big(w_a\coloneqq\begin{bmatrix}
		w_a^{(\gamma)}
	\end{bmatrix}_{|\gamma|<s}:a\in A^m\Big)\in\big(\mb{K}^{\binom{m+s-1}{s-1}}\big)^{|A|^m}.
	\]
	Consider two fresh sets of indeterminates \(\mb{Y}=(Y_\gamma)_{|\gamma|<s}\) and \(\mbf{Z}=(Z_1,\ldots,Z_m)\).  For a suitable parameter \(r\le s\), we will find a nonzero interpolating polynomial with coefficients in \(\mb{K}[\mbf{Z}]\), having the form
	\[
	P(\mb{X},\mb{Y},\mbf{Z})=\wt{P}(\mb{X})+\sum_{j=0}^{r-1}P_j(\mb{X})\bigg(\sum_{|\gamma|=j}Y_\gamma\bigg)\mbf{Z}^\gamma,
	\]
	such that the following hold.
	\begin{enumerate}[(Z1),leftmargin=*]
		\item  The polynomials \(\wt{P}(\mb{X}),P_0(\mb{X}),\ldots,P_{r-1}(\mb{X})\) must have suitably chosen \(\mb{X}\)-degree bounds.  We will also simultaneously ensure a good \(\mbf{Z}\)-degree bound on the coefficients of these polynomials by performing Gaussian elimination over \(\mb{K}[\mbf{Z}]\), as given in~\cite{kannan-1985-linear-system-over-polynomials} or~\cite[Lemma 3]{bhandari-harsha-kumar-sudan-2024-multiplicity-code}.
		\item  For any \(f(\mb{X})\in\mb{K}[\mb{X}]\), define a projection
		\[
		P^{[f]}(\mb{X})\coloneqq\wt{P}(\mb{X})+\sum_{j=0}^{r-1}\bigg(\sum_{|\gamma|=j}\DQ^\gamma f(\mb{X})\bigg)\mbf{Z}^\gamma.
		\]
		Further, for every \(\alpha\in\mb{N}^m,\,|\alpha|<s-r\), define a \(\mb{K}(\mbf{Z})\)-linear operator \(\Delta^{(\alpha)}\) that satisfies the condition \((\Delta^{(\alpha)}(P))^{[f]}(\mb{X})=\DQ^\alpha P^{[f]}(\mb{X})\).  The polynomial \(P(\mb{X},\mb{Y},\mbf{Z})\) must now satisfy the vanishing conditions
		\[
		\Delta^{(\alpha)}(P)(a,w_a,\mbf{Z})=0\qquad\tx{for all }\alpha\in\mb{N}^m\tx{ with }|\alpha|<s-r,\tx{ and }a\in A^m.
		\]
		Note that the vanishing conditions define a linear system on the coefficients of \(P(\mb{X},\mb{Y},\mbf{Z})\), and we obtain the interpolating polynomial as a nontrivial solution of this linear system.
	\end{enumerate}
	So for any close enough polynomial \(f(\mb{X})\) (that is, having a large number of agreements with \(w\)), the above conditions will imply that \(P^{[f]}(\mb{X})\) vanishes at the agreement points with multiplicity at least \(s-r\) (in the sense of multivariate \(\Q\)-derivatives), and this will imply that \(P^{[f]}(\mb{X})=0\) by a suitable version of Polynomial Identity Lemma for multivariate \(\Q\)-derivatives.  We can then solve this equation to recover \(f(\mb{X})\).  Importantly, the list decoding radius is dependent on the choice of \(r\).  We can achieve list decoding up to distance, that is, the claim of~\cref{thm:multivariate-LD} by choosing \(s=O(m^2/\epsilon^2)\) and \(r=O(m^2/\epsilon)\).  In fact, as in the analyses of~\cite{guruswami-wang-2013-FRS,bhandari-harsha-kumar-sudan-2024-multiplicity-code}, we will show that the collection of all close enough polynomials is contained in an affine subspace having dimension at most \(\binom{m+r-2}{r-2}\).
	
	This is a fairly simple strategy, very much within the algebraic algorithmic framework started employed in previous list decoding algorithms~\cite{sudan-1997-RS,guruswami-sudan-1998-RS-list-decoding,guruswami-rudra-2008-FRS,guruswami-wang-2013-FRS,kopparty-2015-multiplicity-code,bhandari-harsha-kumar-sudan-2024-multiplicity-code}.  Further, employing the gluing trick with the additional \(\mbf{Z}\)-variables as in~\cite{bhandari-harsha-kumar-sudan-2024-multiplicity-code} allows us to formulate what we believe is the correct way to extend the linear algebraic decoding framework from the univariate to the multivariate setting, and our extension turns out to be simpler than the analysis in~\cite{bhandari-harsha-kumar-sudan-2024-multiplicity-code}.
	
	There is a subtlety in the recovery step in our multivariate setting vis-\`a-vis the univariate setting.  When we solve the equation \(P^{[f]}(\mb{X})\) to recover a close enough polynomial \(f(\mb{X})\), we will see that the coefficients of \(f(\mb{X})\) will be captured within coefficients of some linear relations between monomials in the \(\mbf{Z}\)-variables.  Using a large enough efficient hitting set (an interpolating set of affine points that is large enough to certify all polynomials with a degree bound) in the \(\mbf{Z}\)-variables, we can then recover each of these coefficients.  This is similar to what was done in~\cite{bhandari-harsha-kumar-sudan-2024-multiplicity-code} in an analogous situation, but our overall analysis is a bit different and simpler.  We will outline this difference now.
	
	Formally, for a vector space \(V_d\subseteq\mb{K}[\mbf{Z}]\) of all polynomials having degree at most \(d\), a \tbf{hitting set} is a finite set \(H\subseteq\mb{K}^m\) such that for any \(g(\mbf{Z})\in V_d\), if \(g(\mbf{Z})\ne0\) then there exists \(u\in H\) such that \(g(u)\ne0\).  For instance, by the Combinatorial Nullstellensatz~\cite[Theorem 1.1]{alon-1999-CN}, every finite grid \(S^m\) with \(S\subseteq\mb{K},\,|S|>d\) is a hitting set for \(V_d\).  (We will precisely use this obvious hitting set in our analysis.)  Our departure from the analysis of~\cite{bhandari-harsha-kumar-sudan-2024-multiplicity-code} is that they need to recover the coefficients of \(f(\mb{X})\) one homogeneous component at a time, by employing the Euler's formula.  In our analysis, we will directly recover one coefficient at a time (with respect to a suitable choice of basis), and this is extremely similar to the corresponding step in the linear algebraic list decoding algorithm~\cite{guruswami-wang-2013-FRS} for classical univariate multiplicity codes.  Recall that we intend to obtain an affine subspace of small dimension that contains all close enough polynomials.  Assume the form \(f(\mb{X})=\sum_{|\alpha|<k}f_\alpha\mb{X}^\alpha\) for all \emph{target polynomials} in the affine subspace.  We will recover the target polynomials inductively from \emph{lower coefficients to higher coefficients}.  In order to kick-start the recovery procedure, as a base case of the induction, we set the coefficients \(\big(f_\gamma:|\gamma|\le r-2\big)\).  Then for any \(\alpha\in\mb{N}^m,\,|\alpha|\le k-r\), the equation \(P^{[f]}(\mb{X})=0\) would imply an equation that has the form
	\begin{align*}
	&\mb{K}(\mb{Z})\tx{-linear combination}\Big(f_{\alpha+\gamma}:|\gamma|=r-1\Big)\\
	&\qquad\qquad\qquad\qquad\qquad=\n\mb{K}(\mb{Z})\tx{-linear combination}\Big(f_{\theta+\eta}:(|\eta|,\theta)\le(r-1,\alpha),\,(|\eta|,\theta)\ne(r-1,\alpha)\Big),
	\end{align*}
	where the polynomials comprising the rational functions (in the \(\mbf{Z}\)-variables) of the above \(\mb{K}(\mbf{Z})\)-linear combinations all have a good \(\mbf{Z}\)-degree bound due to the \(\mbf{Z}\)-degree bound on \(P(\mb{X},\mb{Y},\mbf{Z})\).  Note that each of the coefficients \(f_{\theta+\eta}\) in the R.H.S. above is \emph{strictly below} (in the usual componentwise partial order) some coefficient \(f_{\alpha+\gamma}\) in the L.H.S. above.  So by induction hypothesis, we have already recovered all the coefficients in the R.H.S.  Therefore, the \(\mb{K}(\mbf{Z})\)-linear combination comprising the L.H.S. above is determined.  We then instantiate a hitting set in the \(\mbf{Z}\)-variables to recover the individual coefficients in the L.H.S. above.
	
	So we conclude that the free coefficients \(\big(f_\gamma:|\gamma|\le r-2\big)\) determine the output list polynomials (our target polynomials), and therefore, the output list is contained in an affine \(\mb{K}\)-subspace having dimension at most \(\binom{m+r-2}{r-2}\).
	
	\paragraph*{Constant output list size.}  We can further employ an off-the-shelf \emph{pruning} algorithm by~\cite{kopparty-ron-zewi-saraf-wootters-2023-list-decoding} and its improved analysis by~\cite{tamo-2023-tighter-list-size} to finally end up with a randomized algorithm that guarantees the claimed constant output list size.  This step is nearly identical to that in~\cite{bhandari-harsha-kumar-sudan-2024-multiplicity-code}.  We can also extend the recent argument of~\cite{srivastava-2024-FRS,chen-zhang-2024-FRS-list-size} in an obvious manner to give an even stronger output list size bound (combinatorially).  We will not go into these details in this presentation.

	\subsection*{Further questions}\label{sec:questions}
	
	Multivariate multiplicity codes evaluated on the vector space \(\mb{F}_q^m\) are known to have good locality properties~\cite{kopparty-saraf-yekhanin-2014-multiplicity-codes}, especially when the characteristic is larger than the degree~\cite{kopparty-2015-multiplicity-code,kopparty-ron-zewi-saraf-wootters-2023-list-decoding}.  It will be interesting to see if similar locality properties hold for the multivariate \(\Q\)-multiplicity code in a characteristic insensitive sense.

	\section{Preliminaries}\label{sec:prelims}
	
	\paragraph*{List decoding.}  Let \(\Sigma\) be a finite alphabet.  For any \(a,b\in\Sigma^n\), the \tbf{(relative) Hamming distance} between \(a\) and \(b\) is defined by
	\[
	\msf{d}(a,b)=\frac{1}{n}|\{i\in[n]:a_i\ne b_i\}|.
	\]
	
	Let \(\mb{F}\) be a field, and consider an \(\mb{F}\)-linear code \(C\subseteq\mb{F}^N\) of length \(n\), that is, \(\Sigma=\mb{F}^{N/n}\) is an \(\mb{F}\)-linear space.  For any \(\rho\in[0,1]\), we say \(C\) is \tbf{\((\rho,L)\)-list decodable} if there does not exist any \(w\in\mb{F}^N\) that admits \(L+1\) distinct codewords \(c^{(1)},\ldots,c^{(L+1)}\in C\) such that \(\msf{d}(c^{(j)},w)\le\rho\) for all \(j\in[L+1]\).

	\paragraph*{Field extensions.}  Fix a finite field \(\mb{F}_q\), and a finite degree extension \(\mb{K}/\mb{F}_q\) having extension degree \(\kappa\), that is, \(\mb{K}=\mb{F}_{q^\kappa}\).  Also denote \([\kappa]_q=1+q+\cdots+q^{\kappa-1}\).  Fix a multiplicative generator \(\Q\in\mb{K}^\times\).  We immediately note the following basic observations (see, for instance,~\cite[Chapter 2]{lidl-niederreiter-1997-finite-fields}), and also add quick proofs for completeness.
	\begin{proposition}\label{pro:extension-basic}
		\begin{enumerate}[{\normalfont(a)}]
			\item~{\normalfont\cite[Theorem 2.10]{lidl-niederreiter-1997-finite-fields}}\quad \(\mb{K}=\mb{F}_q(\Q)\).
			\item~{\normalfont\cite[Lemma 2.3]{lidl-niederreiter-1997-finite-fields}}\quad \(\Q^t\not\in\mb{F}_q\), for all \(t\in[[\kappa]_q-1]\).
		\end{enumerate}
	\end{proposition}
	\begin{proof}
		\begin{enumerate}[(a)]
			\item  Since \(\mb{F}_q\subseteq\mb{K}\) and \(\Q\in\mb{K}\), we have \(\mb{F}_q(\Q)\subseteq\mb{K}\).  Since \(\mb{K}^\times=\{\Q^t:t\ge0\}\), we have \(\mb{K}\subseteq\mb{F}_q(\Q)\).
			
			\item  Since \(\Q\) is the multiplicative generator of \(\mb{K}^\times\), and \(\mb{K}=\mb{F}_{q^\kappa}\), the multiplicative order of \(\Q\) is \(q^\kappa-1\).  Obviously, \(\Q^t\ne0\) for all \(t\ge1\).  Now suppose \(\Q^t\in\mb{F}_q^\times\) for some \(t\ge1\).  Then \(\Q^{t(q-1)}=1\), which means \(q^\kappa-1\le t(q-1)\), that is, \(t\ge(q^\kappa-1)/(q-1)=1+q+\cdots+q^{\kappa-1}=[\kappa]_q\).\qedhere
		\end{enumerate}
	\end{proof}
	
	\begin{remark}\label{rem:asymptotic}
	Note that in this work, we will assume that the field size \(q\to\infty\).  We will also consider an additional parameter \(s\ge1\), and always work in the case where \(\Q,\ldots,\Q^{s-1}\not\in\mb{F}_q\).  By~\cref{pro:extension-basic}(b), this is true if \(s-1<[\kappa]_q=1+q+q^2+\cdots+q^{\kappa-1}\).  Further, we will have the degree \(d\) of all polynomials that we consider to satisfy \(d\le sq-1\).  Therefore, we will also need \(sq-1<[\kappa]_q=1+q+q^2+\cdots+q^{\kappa-1}\).  Both these requirements are satisfied if we take \(\kappa=3\) and \(s\le q\), and we assume these throughout the rest of this work.  In fact, we will only take \(s\) to be a constant relative to \(q\).  So henceforth, we have \(\mb{K}=\mb{F}_{q^3}\), and \(\Q\) is a multiplicative generator of \((\mb{F}_{q^3})^\times\).
	\end{remark}

	For any \(n\ge k\ge0\), denote
	\[
	[n]_\Q=\sum_{t=0}^{n-1}\Q^t=\frac{\Q^n-1}{\Q-1},\quad[n]_\Q!=\prod_{t=1}^n[t]_\Q,\quad\tx{and}\quad\Qbinom{n}{k}=\frac{[n]_\Q!}{[k]_\Q![n-k]_\Q!}.
	\]
	Note that \([n]_\Q\ne0\) for all \(n\in[q^\kappa-2]\), and \([0]_\Q!=1\).  The quantity \(\Qbinom{n}{k}\) is called the \tbf{Gaussian binomial coefficient}~\cite{gauss-1808-q-binomial}.  We will also have the convention that \([n]_\Q!=0\) if \(n<0\), and \(\Qbinom{n}{k}=0\) if \(k>n\).

	\subsection{Univariate \(\Q\)-derivatives}\label{sec:uni-Q}
	
	For any \(f(X)\in\mb{K}[X]\), the \tbf{\(\Q\)-derivative}~\cite{jackson-1909-q-taylor-theorem,jackson-1909-q-taylor-series,jackson-1909-q-functions} (also see~\cite{jackson-1910-q-difference-equations,jackson-1910-q-integrals}) is defined by
	\[
	\DQ f(X)=\frac{f(\Q X)-f(X)}{(\Q-1)X}.
	\]
	Further, we are interested in iterated applications of the operator \(\DQ\), and so we denote \(\DQ^0f(X)=f(X)\), and \(\DQ^{t+1}f(X)=\DQ(\DQ^tf)(X)\) for all \(t\ge0\).
	
	\begin{remark}\label{rem:degree-reducing}
		We have \(\DQ^t(X^k)=[k]_\Q[k-1]_\Q\cdots[k-t+1]_\Q X^{k-t+1}\) for all \(k,t\ge0\).  Consider any nonconstant \(f(\mb{X})\in\mb{K}[\mb{X}]\) and \(t\in[0,\deg(f)]\).  Then we immediately get \(\deg(\DQ^t f)\le\deg(f)-t\) by linearity of \(\DQ^t\) (see~\cref{pro:Q-properties}(a) below).  More importantly, as a departure from classical derivatives, we get \(\deg(\DQ^t f)=\deg(f)-t\) if \(\deg(f)\le[\kappa]_q-1\).
	\end{remark}
	
	Let us quickly collect a few basic properties of \(\Q\)-derivatives.  For an indeterminate \(Y\) and \(k\ge0\), denote \((X-Y)_\Q^{(k)}=\prod_{t=0}^{k-1}(X-\Q^tY)\).  For any \(f(X)\in\mb{K}[X]\) and \(a\in\mb{K}^\times\), denote \((f\circ a)(X)=f(aX)\).
	\begin{proposition}[{\cite{kac-cheung-2002-quantum-calculus,ernst-2012-q-calculus}}]\label{pro:Q-properties}
		\begin{enumerate}[{\normalfont(a)}]
			\item\emph{Linearity.}\quad  \(\DQ^k\) is an \(\mb{K}\)-linear map on \(\mb{K}[X]\), for all \(k\ge0\).

			\item\emph{Scaling.}\quad  For any \(f(X)\in\mb{K}[X]\) and \(a\in\mb{K}^\times\),
			\[
			\DQ^k(f\circ a)(X)=a^k\DQ^kf(aX)\quad\tx{for all }k\ge0.
			\]

			\item\emph{Taylor expansion.}\quad  For any \(f(X)\in\mb{K}[X]\),
			\[
			f(X)=\sum_{k\ge0}\frac{\DQ^kf(Y)}{[k]_\Q!}(X-Y)_\Q^{(k)}.
			\]
			In particular,
			\begin{itemize}[leftmargin=*]
				\item[\normalfont\(\bullet\)]  For any \(a\in\mb{K}\), we have \(f(X)=\sum_{k\ge0}\frac{\DQ^kf(a)}{[k]_\Q!}(X-a)_\Q^{(k)}\).
				\item[\normalfont\(\bullet\)]  For any \(k\ge0\), we have \(\DQ^kf(0)=[k]_\Q!\cdot\msf{H}^{(k)}f(0)\), where \(\msf{H}^{(k)}f(0)\) is the \emph{\(k\)-th Hasse derivative}~\footnote{For \(f(X)\in\mb{K}[X]\), the \tbf{Hasse derivatives} \(\msf{H}^{(k)}f(Y),\,k\ge0\) are defined by: \(f(X)=\sum_{k\ge0}\msf{H}^{(k)}f(Y)(X-Y)^k\).} of \(f(X)\) at \(0\).
			\end{itemize}
			
			\item\emph{Product rule.}\quad  For any \(f(X),g(X)\in\mb{K}[X]\) and \(k\ge0\),
			\[
			\DQ^k(fg)(X)=\sum_{t=0}^k\Qbinom{k}{t}\DQ^{k-t}f(\Q^tX)\DQ^tg(X)=\sum_{t=0}^k\Qbinom{k}{t}\DQ^{k-t}f(X)\DQ^tg(\Q^{k-t}X).
			\]
		\end{enumerate}
	\end{proposition}
	
	\paragraph*{Change of basis.}  A simple change of basis shows that \(\Q\)-derivatives are built out of simple evaluations at correlated points.  Towards this, we consider a simple technical lemma, which follows from Gaussian elimination, and extends the criterion for invertibility of finite triangular matrices.
	\begin{lemma}[Instantiation of{~\cite[Theorem 1.2.3]{spiegel-odonnell-1997-incidence-algebras}}]\label{lem:invert-matrix}
		Let \(\mb{F}\) be a field, \(P\) be a ranked poset with rank function \(\rho\).  Let \((r_a)_{a\in P},\,(c_a)_{a\in P}\) be sequences in \(\mb{F}^\times\), and \((t_{a,b})_{a,b\in P}\) be a double sequence in \(\mb{F}\) such that \(t_{a,b}=0\) if \(a\not\ge b\), and \(t_{a,b}\ne0\) if \(a=b\).  Let \(M\in\mb{F}^{P\times P}\) be a lower triangular matrix defined by \(M(a,b)=r_ac_bt_{a,b}\) for all \(a,b\in P\).  Then \(M\) is invertible, and \(M^{-1}\) is lower triangular, given by
		\begin{align*}
		M^{-1}(a,a)&=\frac{1}{r_ac_at_{a,a}}&&\tx{for all }a\in P,\\
		M^{-1}(a,b)&=-\frac{1}{r_bc_at_{b,b}}\sum_{b<z\le a}(-1)^{\rho(a)-\rho(z)}t_{a,z}t_{z,b}&&\tx{for all }a,b\in P,\,b<a.
		\end{align*}
	\end{lemma}
	
	\noindent Define an infinite matrix \(\nu(X)\in(\mb{K}(X))^{\mb{N}\times\mb{N}}\) by
	\[
	\nu_{k,t}(X)=\frac{(-1)^t\Q^{\binom{t}{2}-(k-1)t}}{(\Q-1)^kX^k}\binom{k}{t},\quad\tx{for all }k,t\ge0.
	\]
	Clearly, \(\nu_{k,t}(X)=0\) whenever \(k<t\), that is, \(\nu\) lower triangular.  Further, we have the diagonal entries \(\nu_{k,k}(X)=(-1)^k/(\Q^{\binom{k}{2}}(\Q-1)^kX^k)\ne0\) for all \(k\ge0\).  So by~\cref{lem:invert-matrix}, we conclude that \(\nu\) is invertible, and \(\xi(X)\coloneqq\nu(X)^{-1}\) is lower triangular, given by
	\begin{align*}
		\xi_{k,k}(X)&=(-1)^k\Q^{\binom{k}{2}}(\Q-1)^kX^k\tag*{for all \(k\ge0\),}\\
		\xi_{k,t}(X)&=-(-1)^k\Q^{t(t-1)-\binom{k}{2}}(\Q-1)^tX^t\sum_{u=t+1}^k(-1)^{k-u}\Q^{-(k-1)u-(u-1)t}\binom{k}{u}\binom{u}{t}\qquad\qquad\qquad\qquad\\
		&=-(-1)^k\Q^{-\binom{k}{2}}(\Q-1)^tX^t\sum_{u=t+1}^k(-1)^{k-u}\Q^{-(k-1)u-t(u-t)}\binom{k}{u}\binom{u}{t}\\
		&=(-1)^{t+1}\Q^{-\binom{k}{2}-(k-1)t}(\Q-1)^tX^t\binom{k}{t}\sum_{\ell=1}^{k-t}(-1)^\ell\Q^{-(k-t-1)\ell}\binom{k-t}{\ell}\tag*{since \(\binom{k}{u}\binom{u}{t}=\binom{k}{t}\binom{k-t}{u-t}\)\n\footnotemark\n and \(\ell\coloneqq u-t\)}\\
		&=\frac{(-1)^t(\Q-1)^tX^t}{\Q^{\binom{k}{2}+(k-1)t}}\binom{k}{t}\bigg(1-\bigg(1-\frac{1}{\Q^{k-t-1}}\bigg)^{k-t}\bigg)\tag*{for all \(k>t\ge0\)}
	\end{align*}\footnotetext{See~\cite{graham-knuth-patashnik-1989-concrete-mathematics,greene-knuth-1990-analysis-of-algorithms} for many more extremely useful combinatorial identities.}

	\begin{proposition}[Change of basis for \(\Q\)-derivatives]\label{pro:Q-more-properties}
			For any \(f(X)\in\mb{K}[X]\) and \(k\ge0\), we have
			\[
			\DQ^kf(X)=\sum_{t=0}^k\nu_{k,t}(X)f(\Q^tX)\quad{\tx{\normalfont{\cite{koepf-rajkovic-marinkovic-2007-q-holonomic,annaby-mansour-2008-q-difference}}}},\quad\tx{and}\quad f(\Q^kX)=\sum_{t=0}^k\xi_{k,t}(X)\DQ^tf(X).
			\]
	\end{proposition}
	\begin{proof}
			For any \(f(X)\in\mb{K}[X]\), denote
			\[
			(f\bullet\Q)=\begin{bmatrix}
				f(\Q^kX)
			\end{bmatrix}_k\in(\mb{K}[X])^{\mb{N}\times1}\quad\tx{and}\quad(\DQ\bullet f)=\begin{bmatrix}
				\DQ^kf(X)
			\end{bmatrix}_k\in(\mb{K}[X])^{\mb{N}\times1}.
			\]
			We have \((\DQ\bullet f)=\nu\cdot(f\bullet\Q)\) by~\cite{koepf-rajkovic-marinkovic-2007-q-holonomic,annaby-mansour-2008-q-difference}.  Further, since \(\xi=\nu^{-1}\), this implies \((f\bullet\Q)=\xi\cdot(\DQ\bullet f)\).
	\end{proof}

	\subsection{Multivariate \(\Q\)-derivatives}\label{sec:multi-Q}
	
	Now assume indeterminates \(\mb{X}=(X_1,\ldots,X_m)\).  For every \(i\in[m]\), denote the \(\mb{K}\)-linear \(\Q\)-derivative map on \(\mb{K}[X_i]\) by \(\DQi{X_i}\), and note that it extends in an obvious (and unique) way to a \(\mb{K}(\mb{X}_{[m]\setminus\{i\}})\)-linear map on \(\mb{K}[\mb{X}]\).  It is also immediate that \(\DQi{X_1},\ldots,\DQi{X_m}\) commute as \(\mb{K}\)-linear maps on \(\mb{K}[\mb{X}]\).  Then for any \(f(\mb{X})\in\mb{K}[\mb{X}]\) and \(\alpha\in\mb{N}^m\), we define the \tbf{\(\alpha\)-th \(\Q\)-derivative} by
	\[
	\DQ^\alpha f(\mb{X})=\DQi{X_1}^{\alpha_1}\cdots\DQi{X_m}^{\alpha_m}f(\mb{X}).
	\]
	For \(\alpha,\beta\in\mb{N}^m,\,\beta\le\alpha\), denote
	\[
	[\alpha]_\Q=\prod_{i=1}^m[\alpha_i]_\Q,\quad[\alpha]_\Q!=\prod_{i=1}^m[\alpha_i]_\Q!,\quad\tx{and}\quad\Qbinom{\alpha}{\beta}=\frac{[\alpha]_\Q!}{[\beta]_\Q![\alpha-\beta]_\Q!}=\prod_{i=1}^m\Qbinom{\alpha_i}{\beta_i}.
	\]
	Also, note the usual binomial coefficient \(\binom{\alpha}{\beta}\coloneqq\prod_{i=1}^m\binom{\alpha_i}{\beta_i}\).  For indeterminates \(\mb{Y}=(Y_1,\ldots,Y_m)\) and \(\alpha\in\mb{N}^m\), denote \((\mb{X}-\mb{Y})_\Q^{(\alpha)}=\prod_{i=1}^m(X_i-Y_i)^{(\alpha_i)}_\Q\).  For any \(f(\mb{X})\in\mb{K}[\mb{X}]\) and \(a\in\mb{K}^\times\), denote \(a\mb{X}=(a_1X_1,\ldots,a_mX_m)\), and \((f\circ a)(\mb{X})=f(a\mb{X})\).  Also denote \(a^\gamma=a_1^{\gamma_1}\cdots a_m^{\gamma_m}\).  For any \(\gamma\in\mb{N}^m\), denote \(\Q^\gamma\mb{Y}=(\Q^{\gamma_1}Y_1,\ldots,\Q^{\gamma_m}Y_m)\).
	
	We easily get the following basic properties of multivariate \(\Q\)-derivatives.
	\begin{proposition}\label{pro:multi-Q-properties}
		\begin{enumerate}[{\normalfont(a)}]
			\item\emph{Linearity.}\quad  \(\DQ^\alpha\) is a \(\mb{K}\)-linear map on \(\mb{K}[\mb{X}]\), for all \(\alpha\in\mb{N}^m\).
			\item\emph{Scaling.}\quad  For any \(f(\mb{X})\in\mb{K}[\mb{X}]\) and \(a\in\mb{K}^\times\), we have
			\[
			\DQ^\beta(f\circ a)(\mb{X})=a^\beta\DQ^\beta f(a\mb{X})\quad\tx{for all }\beta\in\mb{N}^m.
			\]
			\item\emph{Taylor expansion.}\quad  For any \(f(\mb{X})\in\mb{K}[\mb{X}]\),
			\[
			f(\mb{X})=\sum_{\alpha\in\mb{N}^m}\frac{\DQ^\alpha f(\mb{Y})}{[\alpha]_\Q!}(\mb{X}-\mb{Y})_\Q^{(\alpha)}.
			\]
			In particular,
			\begin{itemize}[leftmargin=*]
				\item[\normalfont\(\bullet\)]  For any \(a\in\mb{K}^m\), we have \(f(X)=\sum_{\alpha\in\mb{N}^m}\frac{\DQ^\alpha f(a)}{[k]_\Q!}(\mb{X}-a)_\Q^{(\alpha)}\).
				\item[\normalfont\(\bullet\)]  For any \(\alpha\in\mb{N}^m\), we have \(\DQ^\alpha f(0^m)=[\alpha]_\Q!\cdot\msf{H}^{(\alpha)}f(0^m)\), where \(\msf{H}^{(\alpha)}f(0^m)\) is the \emph{\(\alpha\)-th Hasse derivative}~\footnote{For \(f(\mb{X})\in\mb{K}[\mb{X}]\), the \tbf{Hasse derivatives} \(\msf{H}^\alpha f(\mb{Y}),\,\alpha\in\mb{N}^m\) are defined by: \(f(\mb{X})=\sum_{\alpha\in\mb{N}^m}\msf{H}^{(\alpha)}f(\mb{Y})(\mb{X}-\mb{Y})^\alpha\).} of \(f(\mb{X})\) at \(0^m\).
			\end{itemize}
			\item\emph{Product rule.}\quad  For any \(f(\mb{X}),g(\mb{X})\in\mb{K}[\mb{X}]\) and \(\alpha\in\mb{N}^m\),
			\[
			\DQ^\alpha(fg)(\mb{X})=\sum_{\beta\le\alpha}\Qbinom{\alpha}{\beta}\DQ^{\alpha-\beta} f(\Q^\beta\mb{X})\DQ^\beta g(\mb{X})=\sum_{\beta\le\alpha}\Qbinom{\alpha}{\beta}\DQ^{\alpha-\beta} f(\mb{X})\DQ^\beta g(\Q^{\alpha-\beta}\mb{X}).
			\]
		\end{enumerate}
	\end{proposition}
	\begin{proof}
		Each item follows easily by induction on \(m\), with the base cases for the Items (a), (b), (c), and (d)  being~\cref{pro:Q-properties}(a), (b), (c), and (d) respectively.
	\end{proof}

	\paragraph*{Change of basis.}  Define two infinite matrices \(\nu^{(m)}(\mb{X}),\xi^{(m)}(\mb{X}):(\mb{K}(\mb{X}))^{\mb{N}^m\times\mb{N}^m}\) by
	\[
	\nu^{(m)}_{\alpha,\beta}(\mb{X})=\prod_{i=1}^m\nu_{\alpha_i,\beta_i}(X_i),\quad\xi^{(m)}_{\alpha,\beta}(\mb{X})=\prod_{i=1}^m\xi_{\alpha_i,\beta_i}(X_i),\quad\tx{for all }\alpha,\beta\in\mb{N}^m.
	\]
	It is then easy to see (for instance, by induction on \(m\)) that \(\nu^{(m)}(\mb{X})\) is invertible and \(\xi^{(m)}(\mb{X})=(\nu^{(m)}(\mb{X}))^{-1}\).  Also, \(\nu^{(m)}_{\alpha,\beta}(\mb{X})=\xi^{(m)}_{\alpha,\beta}(\mb{X})=0\) for all \(\alpha,\beta\in\mb{N}^m,\,\alpha\not\ge\beta\), and \(\nu^{(m)}_{\alpha,\alpha}(\mb{X})\ne0\ne\xi^{(m)}_{\alpha,\alpha}(\mb{X})\) as well as \(\nu^{(m)}_{\alpha,\alpha}(\mb{X})=1\big/\xi^{(m)}_{\alpha,\alpha}(\mb{X})\) for all \(\alpha\in\mb{N}^m\).

	\begin{proposition}[Change of basis for multivariate \(\Q\)-derivatives]\label{pro:multi-Q-basis-change}
		For any \(f(\mb{X})\in\mb{K}[\mb{X}]\) and \(\alpha\in\mb{N}^m\), we have
		\[
		\DQ^\alpha f(\mb{X})=\sum_{\beta\le\alpha}\nu^{(m)}_{\alpha,\beta}(\mb{X})f(\Q^\beta\mb{X}),\quad\tx{and}\quad f(\Q^\alpha\mb{X})=\sum_{\beta\le\alpha}\xi^{(m)}_{\alpha,\beta}(\mb{X})\DQ^\beta f(\mb{X}).
		\]
	\end{proposition}

	\begin{proof}
		For any \(f(\mb{X})\in\mb{K}[\mb{X}]\), denote
		\[
		(f\bullet\Q)=\begin{bmatrix}
			f(\Q^\alpha\mb{X})
		\end{bmatrix}_\alpha\in(\mb{K}[\mb{X}])^{\mb{N}^m\times1}\quad\tx{and}\quad(\DQ\bullet f)=\begin{bmatrix}
			\DQ^\alpha f(\mb{X})
		\end{bmatrix}_\alpha\in(\mb{K}[\mb{X}])^{\mb{N}^m\times1}.
		\]
		So, we need to prove \((\DQ\bullet f)=\nu^{(m)}\cdot(f\bullet\Q)\) and \((f\bullet\Q)=\xi^{(m)}\cdot(\DQ\bullet f)\).  Since know that \(\xi^{(m)}=(\nu^{(m)})^{-1}\), it is now enough to prove \((\DQ\bullet f)=\nu^{(m)}\cdot(f\bullet\Q)\).  This follows easily by induction on \(m\), with the base case being~\cref{pro:Q-more-properties} (specifically the result by~\cite{koepf-rajkovic-marinkovic-2007-q-holonomic,annaby-mansour-2008-q-difference}).
	\end{proof}
	

	\subsection{Gr\"obner basis, Taylor expansion, and a monomial basis}\label{sec:GB-Taylor}
	
	We will assume familiarity with basic Gr\"obner basis theory.  For definitions and more details, see~\cite[Chapters 2--5]{cox-little-oshea-2013-ideals}.  We will consider vanishing ideals, Gr\"obner bases, and Taylor expansions in the context of multivariate \(\Q\)-derivatives.  These descriptions are quick consequences of the basic properties of multivariate \(\Q\)-derivatives.
	
	Let \(s\ge1\), and consider indeterminates \(\mb{Y}=(Y_1,\ldots,Y_m)\) (that is, a \emph{generic} point in \((\mb{K}(\mb{Y}))^m\)).  By the change of basis given by~\cref{pro:multi-Q-basis-change}, it is immediate that we get a vanishing ideal of \(\{\Q^\gamma\mb{Y}:\gamma\in\mb{N}^m,\,|\gamma|<s\}\) in \(\mb{K}(\mb{Y})[\mb{X}]\) as
	\begin{align*}
		\mc{D}_s^m(\mb{Y})&\coloneqq\{f(\mb{X})\in\mb{K}(\mb{Y})[\mb{X}]:f(\Q^\gamma\mb{Y})=0\tx{ for all }\gamma\in\mb{N}^m,\,|\gamma|<s\}\\
		&=\{f(\mb{X})\in\mb{K}(\mb{Y})[\mb{X}]:\DQ^\gamma f(\mb{Y})=0\tx{ for all }\gamma\in\mb{N}^m,\,|\gamma|<s\}.
	\end{align*}
	Recall that by~\cref{pro:multi-Q-properties}(c), we already have a Taylor expansion given by: For any \(f(\mb{X})\in\mb{K}(\mb{Y})[\mb{X}]\), we have
	\[
	f(\mb{X})=\sum_{\alpha\in\mb{N}^m}\frac{\DQ^\alpha f(\mb{Y})}{[\alpha]_\Q!}(\mb{X}-\mb{Y})_\Q^{(\alpha)}.
	\]
	It then follows immediately that a Gr\"obner basis for the ideal \(\mc{D}^m_s(\mb{Y})\) is
	\[
	\mc{G}^m_s(\mb{Y})\coloneqq\{(\mb{X}-\mb{Y})_\Q^\gamma:|\gamma|=s\}.
	\]
	In fact, this Gr\"obner basis is \emph{universal} (invariant under choice of monomial order) and \emph{reduced} (minimal with respect to divisibility of leading terms).

	As it turns out, the above Gr\"obner basis characterization also extends to finite grids.  For any \(A\subseteq\mb{F}_q^\times\), denote
	\[
	\Q^{(m,s)}(A)=\{\Q^\gamma a:a\in A^m,\,\gamma\in\mb{N}^m,\,|\gamma|<s\}.
	\]
	Let \(\mc{D}_s^m(A)\) denote the vanishing ideal of \(\Q^{(m,s)}(A)\) in \(\mb{K}[\mb{X}]\), that is, again by the change of basis given by~\cref{pro:multi-Q-basis-change},
	\begin{align*}
		\mc{D}_s^m(A)&=\{f(\mb{X})\in\mb{K}[\mb{X}]:f(\Q^\gamma a)=0\tx{ for all }a\in A\tx{ and }\gamma\in\mb{N}^m,\,|\gamma|<s\}\\
		&=\{f(\mb{X})\in\mb{K}[\mb{X}]:\DQ^\gamma f(a)=0\tx{ for all }a\in A\tx{ and }\gamma\in\mb{N}^m,\,|\gamma|<s\}.
	\end{align*}
	The following is a Gr\"obner basis characterization, that is an extension of the Combinatorial Nullstellensatz~\cite[Theorem 1.1]{alon-1999-CN}, and is a special case of~\cite[Theorem 4.7]{geil-martinez-penas-2019}.~\footnote{A related characterization of \emph{standard monomials}, which is a strictly weaker notion than Gr\"obner basis, was obtained much earlier by~\cite{felszeghy-rath-ronya-2006-lex-game,meszaros-2005-diploma-thesis} via determination of winning strategies of a \emph{lex game}, and a further reinterpretation of these winning strategies using the notion of \emph{compression} for set systems.  This reinterpretation was also rediscovered by~\cite{moran-rashtchian-2016-shattering-hilbert-function}.}
	\begin{theorem}[{\cite{geil-martinez-penas-2019}}]\label{thm:grobner-grid}
		For any \(A\subseteq\mb{F}_q^\times\), the set of polynomials
		\[
		\mc{G}^m_s(A)\coloneqq\bigg\{\prod_{i=1}^m\prod_{t=0}^{\gamma_i-1}\bigg(\prod_{a_i\in A}(X_i-\Q^ta_i)\bigg):\gamma\in\mb{N}^m,\,|\gamma|=s\bigg\}
		\]
		is a universal reduced Gr\"obner basis for \(\mc{D}_s^m(A)\).
	\end{theorem}
	\noindent  A similar characterization of Gr\"obner basis for the classical multivariate derivatives was given in~\cite[Section 6]{kopparty-2015-multiplicity-code}, and can also be obtained by~\cite[Theorem 4.7]{geil-martinez-penas-2019}.

	Now note the obvious partial order \(\le\) on \(\mb{N}^m\) defined by \(\alpha\le\beta\) if \(\alpha_i\le\beta_i\) for all \(i\in[m]\).  As before, we will consider finite grids of the form \(A^m\), where \(A\subseteq\mb{F}_q^\times\).  Fix a nonempty set \(A\subseteq\mb{F}_q^\times\), and define the function spaces
	\begin{align*}
		V_{m,s}(A)&=\left\{[f]\coloneqq\begin{bmatrix}
			[f]_a\coloneqq\begin{bmatrix}
				f(\Q^\gamma a)
			\end{bmatrix}_{|\gamma|<s}
		\end{bmatrix}_{a\in A^m}:f(\mb{X})\in\mb{K}[\mb{X}]\right\},\\
		\tx{and}\quad\DQ V_{m,s}(A)&=\left\{[\DQ\mid f]\coloneqq\begin{bmatrix}
			[\DQ\mid f]_a\coloneqq\begin{bmatrix}
				\DQ^\gamma f(a)
			\end{bmatrix}_{|\gamma|<s}
		\end{bmatrix}_{a\in A^m}:f(\mb{X})\in\mb{K}[\mb{X}]\right\}.
	\end{align*}
	By the definitions of the ideal \(\mc{D}^m_s(A)\) in~\cref{sec:GB-Taylor}, it is immediate that \(V_{m,s}(A)\simeq\DQ V_{m,s}(A)\) as \(\mb{K}\)-vector spaces.  Furthermore, the corresponding explicit isomorphism is given by the change of bases in~\cref{pro:multi-Q-basis-change}.  It is also obvious that the above \(\mb{K}\)-vector spaces are isomorphic to \(\big(\mb{K}^{\binom{m+s-1}{s-1}}\big)^{|A|^m}\).
	
	However, the Gr\"obner basis characterization in~\cref{thm:grobner-grid} is even stronger.  As an immediate corollary of these results and the simple properties of Euclidean division for polynomials, we can describe a \emph{monomial basis} (a basis where each vector is the appropriate evaluation of a monomial) for the above vector spaces.  For any \(d\in\mb{N}\), denote \(d^{(m)}=(\underbrace{d,\ldots,d}_{m\tx{ times}})\).  Further for any \(w\in\mb{N}^m\), denote
	\[
	\triangle_k(w)=\bigg\{\alpha\in\mb{N}^m:\left\lfloor\frac{\alpha_1}{w_1}\right\rfloor+\cdots+\left\lfloor\frac{\alpha_m}{w_m}\right\rfloor\le d\bigg\}.
	\]
	\begin{corollary}\label{cor:monomial-basis}
		The collections of evaluation vectors
		\[
		\big\{[\mb{X}^\alpha]:\alpha\in\triangle_{s-1}(|A|^{(m)})\big\}\quad\tx{and}\quad\big\{[\DQ\mid\mb{X}^\alpha]:\alpha\in\triangle_{s-1}(|A|^{(m)})\big\}
		\]
		are \(\mb{K}\)-linear bases of \(V_{m,s}(A)\) and \(\DQ V_{m,s}(A)\) respectively.
	\end{corollary}
	
	\subsection{Counting \(\Q\)-roots of polynomials with multiplicities}\label{sec:counting-roots}
	
	The discussion so far naturally leads us to the question of bounding the number of \emph{\(\Q\)-roots} of a polynomial with multiplicities.  We will only be interested in counting \(\Q\)-roots over the multiplicative group \(\mb{F}_q^\times\).  A tight bound on the number of \(\Q\)-roots is nearly obvious in the case of univariate \(\Q\)-derivatives, and can also be easily obtained in the case of multivariate \(\Q\)-derivatives.  We mention the proofs for completeness.
	
	For any nonzero univariate \(f(X)\in\mb{K}[X]\) and \(a\in\mb{K}\), define the \tbf{\(\Q\)-multiplicity of \(f(X)\) at \(a\)} by
	\[
	\mu_\Q(f,a)=\min\big\{k\ge0:\DQ^kf(a)\ne0\big\}.
	\]
	Note that the \(\Q\)-multiplicity is well-defined, since for any nonzero \(f(X)\in\mb{K}[X]\) and \(a\in\mb{K}\), we have \(\mu_\Q(f,a)\le\deg(f)\) by~\cref{pro:Q-properties}(c).~\footnote{For the zero polynomial, we will fix the convention that \(\deg(0)\coloneqq-\infty\), and \(\mu_\Q(0,a)\coloneqq-\infty\) for all \(a\in\mb{K}\).}
	
	
	\begin{proposition}\label{pro:univariate-SZ}
		For any nonempty set \(A\subseteq\mb{F}_q^\times\), and nonzero \(f(X)\in\mb{K}[X]\) with \(\deg(f)\le d<[\kappa]_q\),~\footnote{If the degree \(d\ge[\kappa]_q\), then the claim of~\cref{pro:univariate-SZ} will change to \(\sum_{a\in A}\min\{\mu_\Q(f,a),[\kappa]_q-1\}\le d\).  We will not need to consider such large degree in our discussion.} we have
		\[
		\sum_{a\in A}\mu_\Q(f,a)\le d.
		\]
		In particular, for any \(s\ge1\), we have \(|\{a\in A:\mu_\Q(f,a)\ge s\}|\le\lfloor d/s\rfloor\).
	\end{proposition}
	\begin{proof}
		For any \(a\in\mb{F}_q^\times\), if \(\mu_\Q(f,a)=r\), then by~\cref{pro:Q-properties}(b), we see that \((X-a)_\Q^{(r)}\) divides \(f(X)\).  Further, by~\cref{pro:extension-basic}(b), for any distinct \(a,a'\in\mb{F}_q\) and \(t\in[[\kappa]_q-1]\), the polynomials \((X-a)_\Q^{(t)},(X-a')_\Q^{(t)}\) are coprime.  Therefore, the product polynomial
		\[
		\prod_{\substack{a\in A\\\mu_\Q(f,a)\ge1}}(X-a)_\Q^{(\min\{\mu_\Q(f,a),[\kappa]_q-1\})}
		\]
		divides \(f(X)\).  This completes the proof.
	\end{proof}

	Now, for any nonzero \(f(\mb{X})\in\mb{K}[\mb{X}]\) and \(a\in\mb{K}^m\), define the \tbf{\(\Q\)-multiplicity of \(f(\mb{X})\) at \(a\)} by
	\[
	\mu_\Q(f,a)=\min\big\{k\ge0:\tx{there exists }\gamma\in\mb{N}^m,\,|\gamma|=k\tx{ such that }\DQ^\gamma f(a)\ne0\big\}.
	\]
	Note that the multivariate \(\Q\)-multiplicity is well-defined, since for any nonzero \(f(\mb{X})\in\mb{K}[\mb{X}]\) and \(a\in\mb{K}^m\), we have \(\mu_\Q(f,a)\le\deg(f)\) by~\cref{pro:multi-Q-properties}(c).~\footnote{For the zero polynomial, we will fix the convention that \(\deg(0)\coloneqq-\infty\), and \(\mu_\Q(0,a)\coloneqq-\infty\) for all \(a\in\mb{K}^m\).}  We need a few quick lemmas, which are consequences of the definition of multivariate \(\Q\)-multiplicity.
	\begin{lemma}\label{lem:multi-Q-multiplicity}
		Consider any nonzero \(f(\mb{X})\in\mb{K}[\mb{X}]\).
		\begin{enumerate}[{\normalfont(a)}]
			\item  For any \(a\in\mb{K}^m\), if \(\mu_\Q(f,a)\ge k\), then
			\[
			\mu_\Q(\DQ^\gamma f,a)\ge k-|\gamma|\quad\tx{for all }\gamma\in\mb{N}^m,\,|\gamma|\le k.
			\]
			
			\item  For any \((a_1,\ldots,a_m)\in\mb{K}^m\) and \(\gamma\in\mb{N}^m\), if \(\DQ^\gamma f(a_1,\ldots,a_{m-1},X_m)\ne0\), then
			\[
			\mu_\Q(\DQ^\gamma f,(a_1,\ldots,a_{m-1},a_m))\le\mu_\Q(\DQ^\gamma f(a_1,\ldots,a_{m-1},X_m),a_m).
			\]
		\end{enumerate}
	\end{lemma}
	\begin{proof}
		\begin{enumerate}[(a)]
			\item  Since  \(\mu_\Q(f,a)\ge k\), we have
			\[
			\DQ^\gamma f(a)=0\quad\tx{for all }\gamma\in\mb{N}^m,\,|\gamma|\le k-1.
			\]
			Consider any \(\gamma,\beta\in\mb{N}^m,\,|\gamma|\le k,\,|\beta|\le k-|\gamma|-1\).  So \(|\beta+\gamma|\le k-1\), and therefore \(\DQ^\beta(\DQ^\gamma f)(a)=\DQ^{\beta+\gamma} f(a)=0\).  This means \(\mu_\Q(\DQ^\gamma f,a)\ge k-|\gamma|\).
			
			\item  Suppose \(\mu_\Q(\DQ^\gamma f,(a_1,\ldots,a_{m-1},a_m))=k\).  This means
			\[
			\DQ^{\beta+\gamma}f(a_1,\ldots,a_{m-1},a_m)=0\quad\tx{for all }\beta\in\mb{N}^m,\,|\beta|\le k-1.
			\]
			In particular, restricting our attention to \(\beta\in\mb{N}^m\) of the form \(\beta=(0^{m-1},t)\), we get
			\[
			\DQi{X_m}^t(\DQ^\gamma f)(a_1,\ldots,a_{m-1},a_m)=0\quad\tx{for all }t\le k-1.
			\]
			This means \(\mu_\Q(\DQ^\gamma f(a_1,\ldots,a_{m-1},X_m),a_m)\ge k\).\qedhere
		\end{enumerate}
	\end{proof}
	
	Quite similar to the classical Polynomial Identity Lemma \cite{ore-1922-zeros,erickson-thesis-1974-polynomial-zeros,demillo-lipton-1978-probabilistic-testing,zippel-1979-probabilistic-algorithms,schwartz-1980-PIT} and its extension to multiplicities by~\cite{dvir-kopparty-saraf-sudan-2013-kakeya}, we can obtain a bound on the number of \(\Q\)-roots of a multivariate polynomial.  Again, we will only be interested in counting \(\Q\)-roots over \(\mb{F}_q^\times\).  The following is a multivariate extension of~\cref{pro:univariate-SZ} for counting \(\Q\)-roots in any grid in \((\mb{F}_q^\times)^m\).  ~\footnote{An elegant alternate proof of the Polynomial Identity Lemma was given by~\cite{moshkovitz-2010-alternative-proof-SZ} in the case where the grid is the vector space \(\mb{F}_q^m\).  This was also extended to multiplicities (for degree at most \(q-1\)) by~\cite{srinivasan-sudan-2019-alternative-proof-mult-SZ}.  For further elegant variant proofs of the Polynomial Identity Lemma, see the answers by Per Vognsen and Arnab Bhattacharyya in the CS Theory Stack Exchange post~\url{https://cstheory.stackexchange.com/questions/1772/alternative-proofs-of-schwartz-zippel-lemma}.}
	\begin{theorem}\label{thm:multivariate-SZ-grid}
		For any nonempty set \(A\subseteq\mb{F}_q^\times\), and any nonzero \(f(\mb{X})\in\mb{K}[\mb{X}]\) with \(\deg(f)\le d<[\kappa]_q\),~\footnote{If the degree \(d\ge[\kappa]_q\), then the claim of~\cref{thm:multivariate-SZ-grid} will change to \(\sum_{a\in A^m}\min\{\mu_\Q(f,a),[\kappa]_q-1\}\le d|A|^{m-1}\).  We will not need to consider such large degree in our discussion.} we have
		\[
		\sum_{a\in A^m}\mu_\Q(f,a)\le d|A|^{m-1}.
		\]
		In particular, for any \(s\ge1\), we have \(|\{a\in A^m:\mu_\Q(f,a)\ge s\}|\le\lfloor d\,|A|^{m-1}/s\rfloor\).
	\end{theorem}
	\begin{proof}
		We will prove by induction on \(m\).  The base case \(m=1\) is true by~\cref{pro:univariate-SZ}.  Now suppose the assertion is true for some \(m\ge1\), and now consider indeterminates \(\mb{X}=(X_1,\ldots,X_m,X_{m+1})\).  Without loss of generality, assume \(\deg(f)=d\), and write
		\[
		f(X_1,\ldots,X_m,X_{m+1})=\sum_{t=0}^\ell f_t(X_1,\ldots,X_m)\cdot X_{m+1}^t,\quad\tx{where }f_\ell(X_1,\ldots,X_m)\ne0,\,\deg(f_\ell)=d-\ell.
		\]
		For any \(a_1,\ldots,a_m\in A\), denote \(s_{a_1,\ldots,a_m}=\mu_\Q(f_\ell,(a_1,\ldots,a_m))\).  So by induction hypothesis, we have
		\begin{align}
			\sum_{(a_1,\ldots,a_m)\in A^m}s_{a_1,\ldots,a_m}\le(d-\ell)|A|^{m-1}.\label{ineq:SZ-induction}
		\end{align}
		Now fix \(a_1,\ldots,a_m\in A\), and let \(\gamma=(\gamma_1,\ldots,\gamma_m)\in\mb{N}^m\) such that \(|\gamma|=s_{a_1,\ldots,a_m}\) and\\\(\DQ^\gamma f_\ell(a_1,\ldots,a_m)\ne0\).  This means
		\begin{itemize}
			\item  \(\DQ^\gamma f_\ell(X_1,\ldots,X_m)\ne0\), and so
			\[
			\DQ^{(\gamma,0)}f(X_1,\ldots,X_m,X_{m+1})=\sum_{t=0}^\ell\DQ^\gamma f_t(X_1,\ldots,X_m)\cdot X_{m+1}^t\ne0.
			\]
			\item  \(g^{(\gamma)}(X_{m+1})\coloneqq\DQ^{(\gamma,0)}f(a_1,\ldots,a_m,X_{m+1})\ne0\), and \(\deg(g^{(\gamma)})=\ell\).
		\end{itemize}
		\noindent Then, for any \(a_{m+1}\in A\), we get
		\begin{align}
			\mu_\Q(f,(a_1,\ldots,a_m,a_{m+1}))&\le|(\gamma,0)|+\mu_\Q(\DQ^{(\gamma,0)} f,(a_1,\ldots,a_m,a_{m+1}))\tag*{by~\cref{lem:multi-Q-multiplicity}(a)}\\
			&\le s_{a_1,\ldots,a_m}+\mu_\Q(g^{(\gamma)},a_{m+1})\tag*{by~\cref{lem:multi-Q-multiplicity}(b)}
		\end{align}
		So by~\cref{pro:univariate-SZ},
		\[
		\sum_{a_{m+1}\in A}\mu_\Q(f,(a_1,\ldots,a_m,a_{m+1}))\le s_{a_1,\ldots,a_m}|A|+\ell.
		\]
		Then (\ref{ineq:SZ-induction}) implies
		\[
		\sum_{(a_1,\ldots,a_m,a_{m+1})\in A^{m+1}}\mu_\Q(f,(a_1,\ldots,a_m,a_{m+1}))\le(d-\ell)|A|^m+\ell|A|^m=d|A|^m,
		\]
		which completes the induction.
	\end{proof}
	
	\subsection{\(\Q\)-multiplicity codes and folded Reed-Muller codes}\label{sec:folded-RM}
	
	Consider \(m,s\ge1\), and any nonempty \(A\subseteq\mb{F}_q^\times\).  For any \(u\in\big(\mb{K}^{\binom{m+s-1}{s-1}}\big)^{|A|^m}\), we will consider the obvious indexing \(u=\begin{bmatrix}
		u_a\coloneqq\begin{bmatrix}
			u_a^{(\gamma)}
		\end{bmatrix}_{|\gamma|<s}
	\end{bmatrix}_{a\in A^m}\).  We will now consider the \tbf{blockwise Hamming distance} \(\msf{d}=\msf{d}^m_{s,A}\) on \(\mb{K}^{\binom{m+s-1}{s-1}\times\binom{m+s-1}{s-1}}\) defined by
	\[
	\msf{d}(u,v)=\frac{|\{a\in A^m:u_a\ne v_a\}|}{|A|^m}.
	\]
	For any \(k\in[s|A|]\), we define the \tbf{\(m\)-variate multiplicity-\(s\) degree-\(k\) \(\Q\)-multiplicity code} by
	\[
	\Qmult_{m,s}(A;k)=\big\{[\DQ\mid f]\in \DQ V_{m,s}(A):\deg(f)<k\big\}.
	\]
	We also denote \(\Qmult_s(A;k)=\Qmult_{1,s}(A;k)\).
	\begin{proposition}
		For any \(k\in[s|A|]\), the code \(\Qmult_{m,s}(A;k)\) has
		\[
		\tx{rate equal to}\n\frac{\binom{m+k-1}{k-1}}{\binom{m+s-1}{s-1}|A|^m},\quad\tx{and distance at least}\n1-\frac{k-1}{s|A|}.
		\]
	\end{proposition}
	\begin{proof}
		Follows immediately from~\cref{cor:monomial-basis} and~\cref{thm:multivariate-SZ-grid}.
	\end{proof}

	Further, for any \(k\in[s|A|]\), we also define the \tbf{\(s\)-folded degree \(k\) folded Reed-Muller (FRM) code} by
	\[
	\FRM_{m,s}(A;k)=\Big\{[\Q\mid f]\coloneqq\begin{bmatrix}
		\begin{bmatrix}
			f(\Q^\gamma a)
		\end{bmatrix}_{|\gamma|<s}
	\end{bmatrix}_{a\in A^m}:f(\mb{X})\in\mb{K}[\mb{X}],\,\deg(f)<k\Big\}.
	\]
	Then, we clearly have the relation
	\[
	\Qmult_{m,s}(A;k)=\diag(U(a):a\in A^m)\cdot\FRM_{m,s}(A;k),
	\]
	where the basis change matrices \(U(a),\,a\in A^m\) are defined by~\cref{pro:multi-Q-basis-change}.  Also note that in the case \(m=1\), \(\FRM_{1,s}(A;k)=\FRS_s(A;k)\) the usual FRS code.

	\section{List decoding of \(\Q\)-multiplicity codes}\label{sec:list-decoding-Q}
	
	We will now move to our main result on list decoding multivariate \(\Q\)-multiplicity codes.  As a warm-up, we will first present the algorithm for the univariate case, which gives an alternate algorithm to list decode the FRS code (after a basis change); this also highlights the similarity with the list decoding algorithm~\cite{guruswami-wang-2013-FRS} for classical univariate multiplicity codes.
	
	\subsection{List decoding of univariate \(\Q\)-multiplicity codes}\label{sec:list-decoding-Q-mult}
	
	There is an obvious efficient list decoding algorithm to list decode the univariate \(\Q\)-multiplicity codes up to capacity, for large folding, which proceeds as follows.  Consider any \(\epsilon>0\), and suppose we have the code \(\Qmult_s(A;k)\) with rate \(R=k/s|A|\), and \(s\sim1/\epsilon^2\).  For any received word \(w\in\mb{K}^{s|A|}\) for the code \(\Qmult_s(A;k)\), consider the tranformed word \(\wt{w}\coloneqq\diag(U(a):a\in A)\cdot w\in\mb{K}^{s|A|}\), which is now a received word for the code \(\FRS_s(A;k)\).  Run the linear algebraic list decoding algorithm~\cite{guruswami-wang-2013-FRS} for the code \(\FRS_s(A;k)\) to get an output solution space of dimension at most \(O(1/\epsilon)\).
	
	
	Let us now see that the classical univariate multiplicity code list decoder~\cite{guruswami-wang-2013-FRS} can also be suitably adapted to list decode \(\Qmult_s(A;k)\).
	\begin{theorem}\label{thm:univariate-LR}
		Consider any \(R\in(0,1),\,\epsilon\in(0,1-R)\).  For any \(A\subseteq\mb{F}_q^\times,\,|A|=n\), and for the choices \(s=\lceil1/\epsilon^2\rceil\) and \(k=\lceil Rsn\rceil\), the code \(\Qmult_s(A;k)\) is efficiently list decodable up to radius \(1-R-\epsilon\) with output list contained in a \(\mb{K}\)-affine space of dimension at most \(O(1/\epsilon)\).
	\end{theorem}
	
	Consider additional indeterminates \(\mb{Y}=(Y_0,\ldots,Y_{s-1})\), and the \(\mb{K}\)-linear subspace of the polynomial ring \(\mb{K}[X,\mb{Y}]\) defined by
	\[
	\msf{L}_s(X,\mb{Y})=\big\{\wt{P}(X)+P_0(X)Y_0+\cdots+P_{s-1}(X)Y_{s-1}:\wt{P}(X),P_0(X),\ldots,P_{s-1}(X)\in\mb{K}[X]\big\}.
	\]
	Also, denote \(Y_s=0\).  Define a \(\mb{K}\)-linear operator \(\Delta:\msf{L}_s(X,\mb{Y})\to\msf{L}_s(X,\mb{Y})\) by
	\[
	\Delta\bigg(\wt{P}(X)+\sum_{j=0}^{s-1}P_j(X)Y_j\bigg)=\DQ\wt{P}(X)+\sum_{j=0}^{s-1}\big(\DQ P_j(X)Y_j+P_j(\Q X)Y_{j+1}\big).
	\]
	We will be interested in iterated applications of \(\Delta\), and therefore, denote \(\Delta^0=\mrm{Id}\) and \(\Delta^{r+1}=\Delta\circ\Delta^r\) for all \(r\ge0\).  For any \(P(X,\mb{Y})\in\msf{L}_s(X,\mb{Y})\) and \(f(X)\in\mb{K}[X]\), denote \(P^{[f]}(X)=P(X,f(X),\DQ f(X),\ldots,\DQ^{s-1}f(X))\).  The following is immediate by the definition of \(\Delta\).
	\begin{lemma}\label{lem:iterated-Delta}
		Let \(P(X,\mb{Y})\in\msf{L}_s(X,\mb{Y})\).
		\begin{enumerate}[{\normalfont(a)}]
			\item  If \(j\in[0,s-1]\) such that \(\deg_{Y_{j'}}(P)=0\) for all \(j'\ge j\), then \(\deg_{Y_{j'}}(\Delta P)=0\) for all \(j'\ge j+1\).
			\item  \((\Delta P)^{[f]}(X)=\DQ(P^{[f]})(X)\), for all \(f(X)\in\mb{K}[X]\).
		\end{enumerate}
	\end{lemma}
	\begin{proof}
		Let \(P(X,\mb{Y})=\wt{P}(X)+P_0(X)Y_0+\cdots+P_{s-1}(X)Y_{s-1}\in\msf{L}_s(X,\mb{Y})\).
		\begin{enumerate}[(a)]
			\item  If \(j\in[0,s-1]\) such that \(\deg_{Y_j'}(P)=0\) for all \(j'\ge j\), this means \(P_j(X)=\cdots=P_{s-1}(X)=0\).  Now we have
			\begin{align*}
			(\Delta P)(X,\mb{Y})&=\DQ\wt{P}(X)+\sum_{j'=0}^{s-1}\big(\DQ P_{j'}(X)Y_j+P_{j'}(\Q X)Y_{j'+1}\big)\\
			&=\DQ\wt{P}(X)+\sum_{j'=0}^{s-1}\big(P_{j'-1}(\Q X)+\DQ P_{j'}(X)\big)Y_{j'}.
			\end{align*}
			So, \(P_{j'-1}(\Q X)+\DQ P_{j'}(X)=0\) for all \(j'\ge j+1\).  This means \(\deg_{Y_{j'}}(\Delta P)=0\) for all \(j'\ge j+1\).
			
			\item  We have
			\begin{align}
			(\Delta P)^{[f]}(X)&=\DQ\wt{P}(X)+\sum_{j=0}^{s-1}\big(\DQ P_j(X)\DQ^j f(X)+P_j(\Q X)\DQ^{j+1}f(X)\big)\notag\\
			&=\DQ\wt{P}(X)+\sum_{j=0}^{s-1}\DQ\big(P_j\cdot\DQ^jf\big)(X)\tag*{by~\cref{pro:Q-properties}(c)}\\
			&=\DQ P^{[f]}(X).\notag\qedhere
			\end{align}
		\end{enumerate}
	\end{proof}
	
	It is worth noting that a close variant of the operator \(\Delta\) appears in~\cite[Definition 8.1]{goyal-kumar-harsha-shankar-2023-fast} in the context of nearly linear-time list decoding of FRS codes, without the interpretation in terms of \(\Q\)-derivatives.  So it is reasonable to surmise that a nearly linear-time implementation of the list decoding algorithm \`a la~\cite{goyal-kumar-harsha-shankar-2023-fast} is possible for the univariate \(\Q\)-multiplicity codes.  In order to keep our main presentation short, we limit ourselves here to adapting the more conventional polynomial-time algorithm of~\cite{guruswami-wang-2013-FRS}.
	
	The interpolation step of the list recovery algorithm, which will be used multiple times in this discussion, can be captured as follows.
	\begin{proposition}\label{pro:LR-interpolation}
		Consider any \(A=\{\alpha_1,\ldots,\alpha_n\}\subseteq\mb{F}_q^\times\), parameters \(s\ge r\ge1\), and \(k\in[s|A|]\).  Define
		\[
		d=\left\lceil\frac{n(s-r+1)-(r+k)+1}{r+1}\right\rceil.
		\]
		Then for any \(w=(w_1,\ldots,w_n)\in(\mb{K}^s)^n\), there exists a nonzero polynomial \(P(X,\mb{Y})=\wt{P}(X)+\sum_{j=0}^{r-1}P_j(X)Y_j\in\msf{L}_s(X,\mb{Y})\) with
		\[
		\deg(\wt{P})\le d+k-1,\quad\tx{and}\quad\deg(P_j)\le d\quad\tx{for all }j\in[0,r-1],
		\]
		such that for any \(f(X)\in\mb{K}[X],\,\deg(f)<k\),
		\[
		\tx{if}\quad\msf{d}^{(s)}(f,w)\ge\frac{1}{r+1}+\frac{k}{(s-r+1)n}\qquad\tx{then}\quad P(X,f(X),\DQ f(X),\ldots,\DQ^{s-1}f(X))=0.
		\]
	\end{proposition}
	\begin{proof}
		Assume the notation \(u=(u^{(0)},\ldots,u^{(s-1)})\in\mb{K}^s\).  We will construct a nonzero polynomial \(P(X,\mb{Y})=\wt{P}(X)+\sum_{j=0}^{s-1}P_j(X)Y_j\in\msf{L}_s(X,\mb{Y})\) such that
		\begin{enumerate}[(a)]
			\item  \(P(X,\mb{Y})=\wt{P}(X)+\sum_{j=0}^{r-1}P_j(X)Y_j\in\msf{L}_s(X,\mb{Y})\), that is, \(P_r(X)=\cdots=P_{s-1}(X)=0\).
			\item  \(\deg(\wt{P})\le d+k-1\), and \(\deg(P_j)\le d\) for all \(j\in[0,r-1]\).
			\item  \((\Delta^jP)\big(\alpha_i,w_i^{(0)},\ldots,w_i^{(s-1)}\big)=0\) for all \(j\in[0,s-r],\,i\in[n]\).
		\end{enumerate}
		The number of linear constraints is \(n(s-r+1)\), and the number of coefficients is \(d+k+r(d+1)=d(r+1)+(r+k)\).  So for the choice
		\[
		d=\left\lceil\frac{n(s-r+1)-(r+k)+1}{r+1}\right\rceil,
		\]
		we indeed have \(d(r+1)+(r+k)>n(s-r+1)\), and this ensures a nontrivial solution, that is, \(P(X,\mb{Y})\ne0\).
		
		Now consider any \(f(X)\in\mb{K}[X]\), and suppose \(\msf{d}^{(s)}(f,w)=1-(\nu/n)\).  By~\cref{lem:iterated-Delta}, this implies \(\sum_{a\in A}\mu_\Q(P^{[f]},a)\ge\nu s-r+1\).  But we also have \(\deg(P^{[f]})\le d+k-1\).  Therefore, by~\cref{pro:univariate-SZ}, we can conclude that \(P^{[f]}(X)=0\) if \(\nu>\lceil (d+k-1)/(s-r+1)\rceil\).  This is equivalent to
		\begin{align*}
			\nu>\left\lceil\frac{\left\lceil\frac{n(s-r+1)-(r+k)+1}{r+1}\right\rceil+k-1}{s-r+1}\right\rceil=\left\lceil\frac{n(s-r+1)-(r+k)+(r+1)(k-1)+1}{(r+1)(s-r+1)}\right\rceil.
		\end{align*}
		The above is true if
		\[
		\nu\ge\frac{n}{r+1}+\frac{k}{s-r+1},\quad\tx{that is,}\quad\frac{\nu}{n}\ge\frac{1}{r+1}+\frac{k}{(s-r+1)n}.\qedhere
		\]
	\end{proof}
	
	Once we have an interpolating polynomial that captures the correct codewords, we can extract the output list as follows.
	\begin{proposition}\label{pro:LR-solving}
		Consider parameters \(s\ge r\ge1\), and \(k,d\ge1\) such that \(d+k-1<[3]_q\).  If \(P(X,\mb{Y})=\wt{P}(X)+\sum_{j=0}^{r-1}P_j(X)Y_j\in\msf{L}_s(X,\mb{Y})\) is a nonzero polynomial with
		\[
		\deg(\wt{P})\le d+k-1,\quad\tx{and}\quad\deg(P_j)\le d\quad\tx{for all }j\in[0,r-1],
		\]
		then the solution space
		\[
		\{f(X)\in\mb{K}[X]:P^{[f]}(X)=0\}
		\]
		is a \(\mb{K}\)-affine space with dimension at most \(r-1\).
	\end{proposition}
	\begin{proof}
		Consider any \(f(X)\in\mb{K}[X]\) satisfying \(P^{[f]}(X)=0\).  Note that we have \(P(X,\mb{Y})\ne0\).  If \(P_j(X)=0\) for all \(j\in[0,r-1]\), then \(P^{[f]}(X)=\wt{P}(X)=0\), which means \(P(X,\mb{Y})=0\), a contradiction.  So there exists \(j\in[0,r-1]\) such that \(P_j(X)\ne0\).  Without loss of generality, we assume \(P_{r-1}(X)\ne0\).  (Otherwise, we work with the largest \(r'\) such that \(P_{r'-1}(X)\ne0\).)
		
		Let us also quickly consider another notation.  For any \(\Q^b\in\mb{K}^\times,\,h\ge0\), and \(g(X)\in\mb{K}[X]\), denote
		\begin{align*}
			g_{h,b}\coloneqq\coeff\bigg(\frac{(X-\Q^b)_\Q^{(h)}}{[h]_\Q!},g\bigg)=\DQ^hg(\Q^b)\tag*{by~\cref{pro:Q-properties}(b)}
		\end{align*}
		Recall that \(\mb{K}=\mb{F}_{q^3}\).  Since \(\deg(P_{r-1})\le d<[3]_q\), there exists \(\Q^b\in\mb{K}^\times\) such that \(P_{r-1}(\Q^{h+b})\ne0\) for all \(h\in[0,d+k-1]\).  Further, since \(\deg(P)\le d+k-1<[3]_q\), we will work with the basis of monomials \(\{(X-\Q^b)_\Q^{(h)}:h\in[0,d+k-1]\}\).  Now we have
		\[
		P^{[f]}(X)\equiv\wt{P}(X)+P_0(X)f(X)+P_1(X)\DQ f(X)+\cdots+P_{r-1}(X)\DQ^{r-1}f(X)=0
		\]
		So by~\cref{pro:Q-properties}(c), for every \(h\in[0,d+k-1]\), we have
		\begin{align*}
			0=P^{[f]}_{h,b}&=(\wt{P})_{h,b}+\sum_{j=0}^{r-1}\sum_{c=0}^h\Qbinom{h}{c}(P_j)_{h-c,c+b}(\DQ^jf)_{c,b}\\
			&=(\wt{P})_{h,b}+\sum_{j=0}^{r-1}\sum_{c=0}^h\Qbinom{h}{c}(P_j)_{h-c,c+b}f_{c+j,b}.
		\end{align*}
		Since \((P_{r-1})_{0,h+b}=P_{r-1}(\Q^{h+b})\ne0\) for all \(h\in[0,d+k-1]\), we can rewrite the above as
		\[
		f_{h+r-1,b}=-\frac{1}{(P_{r-1})_{0,h+b}}\Bigg((\wt{P})_{h,b}+\sum_{\substack{(j,c)\le(r-1,h)\\(j,c)\ne(r-1,h)}}\Qbinom{h}{c}(P_j)_{h-c,c+b}f_{c+j,b}\Bigg)\quad\tx{for all }h\in[0,d+k-1].
		\]
		So we can set the undetermined coefficients among \(f_{0,b},\ldots,f_{r-2,b}\in\mb{K}\) freely, and the remaining coefficients of \(f(X)\) are uniquely determined the above equation.  This means the solution space is a \(\mb{K}\)-affine space having dimension at most \(r-1\).
	\end{proof}
	
	Our main theorem is now immediate.
	\begin{proof}[Proof of Theorem~\ref{thm:univariate-LR}]
		Choosing \(r=O(1/\epsilon)\), and plugging in all the parameters into~\cref{pro:LR-interpolation} and~\cref{pro:LR-solving} immediately implies the claim.
	\end{proof}

	\subsection{List decoding of multivariate \(\Q\)-multiplicity codes}\label{sec:list-decoding-multi-Q-mult}
	
	We will now prove our main result (\cref{thm:multivariate-LD}) on list decoding \(\Qmult_{m,s}(A;k)\).  Even though it is inspired by the classical multivariate multiplicity code list decoder~\cite{bhandari-harsha-kumar-sudan-2024-multiplicity-code}, it turns out that we can give a conceptually simpler algorithm and analysis that is perhaps the correct multivariate extension of our univariate list decoder (\cref{thm:univariate-LR}).
	
	Consider additional indeterminates \(\mb{Y}=(Y_\gamma:|\gamma|<s),\,\mbf{Z}=(Z_1,\ldots,Z_m)\), and the \(\mb{K}(\mbf{Z})\)-linear subspace of the polynomial ring \(\mb{K}(\mbf{Z})[\mb{X},\mb{Y}]\) defined by
	\[
	\msf{L}_s(\mb{X},\mb{Y})=\bigg\{\wt{P}(\mb{X})+\sum_{j=0}^{s-1}P_j(\mb{X})\sum_{|\gamma|=j}Y_\gamma\mbf{Z}^\gamma\in\mb{K}(\mbf{Z})[\mb{X},\mb{Y}]\bigg\}.
	\]
	Also, denote \(Y_\gamma=0\) for all \(\gamma\in\mb{N}^m,\,|\gamma|\ge s\).  For any \(\alpha\in\mb{N}^m\), define a \(\mb{K}(\mbf{Z})\)-linear operator \(\Delta^{(\alpha)}:\msf{L}_s(\mb{X},\mb{Y})\to\msf{L}_s(\mb{X},\mb{Y})\) by
	\[
	\Delta^{(\alpha)}\bigg(\wt{P}(\mb{X})+\sum_{j=0}^{s-1}P_j(\mb{X})\sum_{|\gamma|=j}Y_\gamma\mbf{Z}^\gamma\bigg)=\DQ^\alpha\wt{P}(\mb{X})+\sum_{j=0}^{s-1}\sum_{|\gamma|=j}\bigg(\sum_{\beta\le\alpha}\Qbinom{\alpha}{\beta}\DQ^{\alpha-\beta}P_j(\Q^\beta\mb{X})Y_{\beta+\gamma}\bigg)\mbf{Z}^\gamma.
	\]
	For any \(P(\mb{X},\mb{Y},\mbf{Z})\in\msf{L}_s(\mb{X},\mb{Y})\) and \(f(\mb{X})\in\mb{K}[\mb{X}]\), denote \(P^{[f]}(\mb{X})=P\big(\mb{X},\big(\DQ^\gamma f(\mb{X}):|\gamma|<s\big),\mbf{Z}\big)\in\mb{K}(\mbf{Z})[\mb{X}]\).  The following is then immediate.
	\begin{lemma}\label{lem:-multi-iterated-Delta}
		Let \(P(\mb{X},\mb{Y},\mbf{Z})\in\msf{L}_s(\mb{X},\mb{Y})\).
		\begin{enumerate}[{\normalfont(a)}]
			\item  If \(j\in[0,s-1]\) such that \(\deg_{Y_\gamma}(P)=0\) for all \(\gamma\in\mb{N}^m,\,|\gamma|\ge j\), then \(\deg_{Y_\gamma}(\Delta^{(\alpha)}P)=0\) for all \(\gamma\in\mb{N}^m,\,|\gamma|\ge j+|\alpha|\).
			\item  \((\Delta^{(\alpha)}P)^{[f]}(\mb{X})=\DQ^\alpha(P^{[f]})(\mb{X})\), for all \(f(\mb{X})\in\mb{K}[\mb{X}]\).
		\end{enumerate}
	\end{lemma}
	\begin{remark}
		A consequence of~\cref{lem:-multi-iterated-Delta} is that \(\Delta^{(\alpha+\beta)}=\Delta^{(\alpha)}\circ\Delta^{(\beta)}\) for all \(\alpha,\beta\in\mb{N}^m\).  We do not explicitly use this elsewhere.
	\end{remark}
	\begin{proof}[Proof of~\cref{lem:-multi-iterated-Delta}]
		Let \(P(\mb{X},\mb{Y},\mbf{Z})=\wt{P}(\mb{X})+\sum_{j=0}^{s-1}P_j(\mb{X})\sum_{|\gamma|=j}Y_\gamma\mbf{Z}^\gamma\in\msf{L}_s(\mb{X},\mb{Y})\).
		\begin{enumerate}[(a)]
			\item  If \(j\in[0,s-1]\) such that \(\deg_{Y_\gamma}(P)=0\) for all \(\gamma\in\mb{N}^m,\,|\gamma|\ge j\), this means \(P_{j'}(\mb{X})=0\) for all \(j'\ge j\).  Now we have
			\begin{align*}
			(\Delta^{(\alpha)}P)(\mb{X},\mb{Y},\mbf{Z})&=\DQ^\alpha\wt{P}(\mb{X})+\sum_{j'=0}^{s-1}\sum_{|\gamma|=j'}\bigg(\sum_{\beta\le\alpha}\Qbinom{\alpha}{\beta}\DQ^{\alpha-\beta}P_{j'}(\Q^\beta\mb{X})Y_{\beta+\gamma}\bigg)\mbf{Z}^\gamma\\
			&=\DQ^\alpha\wt{P}(\mb{X})+\sum_{j'=0}^{s-1}\sum_{|\gamma|=j'}\bigg(\sum_{\beta\le\alpha}\Qbinom{\alpha}{\beta}\DQ^{\alpha-\beta}P_{j'-|\beta|}(\Q^\beta\mb{X})\mbf{Z}^{\gamma-\beta}\bigg)Y_\gamma
			\end{align*}
			So, \(\sum_{\beta\le\alpha}\DQ^{\alpha-\beta}P_{j'-|\beta|}(\Q^\beta\mb{X})\mbf{Z}^{\gamma-\beta}=0\) for all \(j'\ge j+|\alpha|\).  This means \(\deg_{Y_{j'}}(\Delta^{(\alpha)}P)=0\) for all \(j'\ge j+|\alpha|\).
			
			\item  We have
			\begin{align*}
				(\Delta^{(\alpha)}P)^{[f]}(\mb{X})&=\DQ^\alpha\wt{P}(\mb{X})+\sum_{j=0}^{s-1}\sum_{|\gamma|=j}\bigg(\sum_{\beta\le\alpha}\Qbinom{\alpha}{\beta}\DQ^{\alpha-\beta}P_{j}(\Q^\beta\mb{X})\DQ^{\beta+\gamma}f(\mb{X})\bigg)\mbf{Z}^\gamma\\
				&=\DQ^\alpha\wt{P}(\mb{X})+\sum_{j=0}^{s-1}\sum_{|\gamma|=j}\DQ^\alpha\big(P_j\cdot\DQ^\gamma f\big)(\mb{X})\mbf{Z}^\gamma\tag*{by~\cref{pro:multi-Q-properties}(d)}\\
				&=\DQ^\alpha(P^{[f]})(\mb{X}).\qedhere
			\end{align*}
		\end{enumerate}
	\end{proof}
	
	The interpolation step of the list recovery algorithm, which will be used multiple times in this discussion, can be captured as follows.
	\begin{proposition}\label{pro:multi-LR-interpolation}
		Consider any \(A\subseteq\mb{F}_q^\times\), parameters \(s\ge r\ge1\), and \(k\in[s|A|]\).  Define
		\[
		d=\left\lceil5\bigg(\frac{1}{r+1}\bigg)^{1/m}(s-r+1)|A|\right\rceil.
		\]
		Then for any \(w=(w_a:a\in A^m)\in\big(\mb{K}^{\binom{m+s-1}{s-1}}\big)^{|A|^m}\), there exists a nonzero polynomial \(P(\mb{X},\mb{Y},\mbf{Z})=\wt{P}(\mb{X})+\sum_{j=0}^{r-1}P_j(\mb{X})\sum_{|\gamma|=j}Y_\gamma\mbf{Z}^\gamma\in\msf{L}_s(\mb{X},\mb{Y})\) with
		\[
		\deg(\wt{P})\le d+k-1,\quad\tx{and}\quad\deg(P_j)\le d\quad\tx{for all }j\in[0,r-1].
		\]
		such that for any \(f(\mb{X})\in\mb{K}[\mb{X}],\,\deg(f)<k\),
		\[
		\tx{if}\quad\msf{d}^{(m,s)}(f,w)\ge5\bigg(\frac{1}{r+1}\bigg)^{1/m}+\frac{k}{(s-r+1)|A|^{m-1}}.\qquad\tx{then}\quad P^{[f]}(\mb{X})=0.
		\]
		Moreover, we can ensure that \(P(\mb{X},\mb{Y},\mbf{Z})\) is a polynomial in \(\mbf{Z}\) with \(\deg_{\mbf{Z}}(P)<ms\).
	\end{proposition}
	\begin{proof}
		Assume the notation \(u=(u^{(\gamma)}:|\gamma|<s)\in\mb{K}^{\binom{m+s-1}{s-1}}\).  We will construct a nonzero polynomial \(P(\mb{X},\mb{Y},\mbf{Z})=\wt{P}(\mb{X})+\sum_{j=0}^{s-1}P_j(\mb{X})\sum_{|\gamma|=j}Y_\gamma\mbf{Z}^\gamma\in\msf{L}_s(\mb{X},\mb{Y})\) such that
		\begin{enumerate}[(a)]
			\item  \(P(\mb{X},\mb{Y},\mbf{Z})=\wt{P}(\mb{X})+\sum_{j=0}^{r-1}P_j(\mb{X})\sum_{|\gamma|=j}Y_\gamma\mbf{Z}^\gamma\in\msf{L}_s(\mb{X},\mb{Y})\), that is, \(P_r(\mb{X})=\cdots=P_{s-1}(\mb{X})=0\).
			\item  \(\deg(\wt{P})\le d+k-1\), and \(\deg(P_j)\le d\) for all \(j\in[0,r-1]\).
			\item  \((\Delta^{(\alpha)}P)\big(a,(w_a^{(\gamma)}:|\gamma|<s)\big)=0\) for all \(\alpha\in\mb{N}^m\) with \(|\alpha|\in[0,s-r],\,a\in A^m\).
		\end{enumerate}
		The number of linear constraints is
		\[
		|A|^m\binom{m+s-r}{s-r}\le\bigg(\frac{5(s-r+1)|A|}{m}\bigg)^m,
		\]
		and the number of coefficients is
		\[
		\binom{m+d+k-1}{d+k-1}+r\binom{m+d}{d}\ge(r+1)\binom{m+d}{d}\ge(r+1)\bigg(\frac{d}{m}\bigg)^m.
		\]
		So for the choice
		\[
		d=\left\lceil5\bigg(\frac{1}{r+1}\bigg)^{1/m}(s-r+1)|A|\right\rceil,
		\]
		we indeed get a nontrivial solution, that is, \(P(\mb{X},\mb{Y},\mbf{Z})\ne0\).  Further, since each linear constraint is a polynomial in \(\mbf{Z}=(Z_1,\ldots,Z_m)\) having degree at most \(s-1\), we can ensure that \(P(\mb{X},\mb{Y},\mbf{Z})\) is a polynomial in \(\mbf{Z}\) having \(\deg_{\mbf{Z}}(P)<ms\) (see~\cite{kannan-1985-linear-system-over-polynomials} or~\cite[Lemma 3]{bhandari-harsha-kumar-sudan-2024-multiplicity-code}.)
		
		Now consider any \(f(X)\in\mb{K}[X]\), and suppose \(\msf{d}^{(m,s)}(f,S)=1-(\nu/|A|^{m-1})\).  By~\cref{lem:-multi-iterated-Delta}, this implies \(\sum_{a\in A^m}\mu_\Q(P^{[f]},a)\ge\nu(s-r+1)\).  But we also have \(\deg(P^{[f]})\le(d+k-1)|A|^{m-1}\).  Therefore, by~\cref{thm:multivariate-SZ-grid}, we can conclude that \(P^{[f]}(\mb{X})=0\) if \(\nu>(d+k-1)|A|^{m-1}/(s-r+1)\), which holds if
		\[
		\nu>\left\lceil\frac{5(1/(r+1))^{1/m}(s-r+1)|A|^{m-1}+k-1}{s-r+1}\right\rceil.
		\]
		The above is true if
		\[
		\nu\ge5\bigg(\frac{1}{r+1}\bigg)^{1/m}|A|^{m-1}+\frac{k}{s-r+1},\quad\tx{that is,}\quad\frac{\nu}{|A|^{m-1}}\ge5\bigg(\frac{1}{r+1}\bigg)^{1/m}+\frac{k}{(s-r+1)|A|^{m-1}}.\qedhere
		\]
	\end{proof}

	Once we have an interpolating polynomial that captures the correct codewords, we can extract the output list as follows.
	\begin{proposition}\label{pro:multi-LR-solving}
		Consider parameters \(s\ge r\ge1\), and \(k,d\ge1\) such that \(d+k-1<[3]_q\).  If \(P(\mb{X},\mb{Y},\mbf{Z})=\wt{P}(\mb{X})+\sum_{j=0}^{r-1}P_j(\mb{X})\sum_{|\gamma|=j}Y_\gamma\mbf{Z}^\gamma\in\msf{L}_s(\mb{X},\mb{Y},\mbf{Z})\) is a nonzero polynomial with
		\[
		\deg(\wt{P})\le d+k-1,\quad\tx{and}\quad\deg(P_j)\le d\quad\tx{for all }j\in[0,r-1],
		\]
		then the solution space
		\[
		\{f(\mb{X})\in\mb{K}[\mb{X}]:P^{[f]}(\mb{X})=0\}
		\]
		is a \(\mb{K}\)-affine space with dimension at most \(\binom{m+r-2}{r-2}\).
	\end{proposition}
	\begin{proof}
		Consider any \(f(\mb{X})\in\mb{K}[\mb{X}]\) satisfying \(P^{[f]}(\mb{X})=0\).  Note that we have \(P(\mb{X},\mb{Y},\mb{Z})\ne0\).  If \(P_j(\mb{X})=0\) for all \(j\in[0,r-1]\), then \(P^{[f]}(\mb{X})=\wt{P}(\mb{X})=0\), which means \(P(\mb{X},\mb{Y},\mbf{Z})=0\), a contradiction.  So there exists \(j\in[0,r-1]\) such that \(P_j(\mb{X})\ne0\).  Without loss of generality, we assume \(P_{r-1}(\mb{X})\ne0\).  (Otherwise, we work with the largest \(r'\) such that \(P_{r'-1}(\mb{X})\ne0\).)
		
		Let us also quickly consider another notation.  For any \(\Q^\beta\in(\mb{K}^\times)^m,\,\alpha\in\mb{N}^m\), and \(g(\mb{X})\in\mb{K}[\mb{X}]\), denote
		\begin{align*}
			g_{\alpha,\beta}\coloneqq\coeff\bigg(\frac{(\mb{X}-\Q^\beta)_\Q^{(\alpha)}}{[\alpha]_\Q!},g\bigg)=\DQ^\alpha g(\Q^\beta)\tag*{by~\cref{pro:multi-Q-properties}(b)}
		\end{align*}
		Recall that \(\mb{K}=\mb{F}_{q^3}\).  Since \(\deg(P_{r-1})\le d<[3]_q\), there exists \(\Q^\beta\in(\mb{K}^\times)^m\) such that \(P_{r-1}(\Q^{\alpha+\beta})\ne0\) for all \(\alpha\in\mb{N}^m,\,|\alpha|\le d+k-1\).  Further, since \(\deg(P)\le d+k-1<[3]_q\), we will work with the basis of monomials \(\{(\mb{X}-\Q^\beta)_\Q^{(\alpha)}:\alpha\in\mb{N}^m,\,|\alpha|\le d+k-1\}\).  Now we have
		\[
		P^{[f]}(\mb{X})\equiv\wt{P}(\mb{X})+\sum_{j=0}^{r-1}P_j(\mb{X})\sum_{|\gamma|=j}\DQ^\gamma f(\mb{X})\mbf{Z}^\gamma=0.
		\]
		So for every \(\alpha\in\mb{N}^m,\,|\alpha|\le d+k-1\), we have
		\begin{align*}
			0=P^{[f]}_{\alpha,\beta}&=\wt{P}_{\alpha,\beta}+\sum_{j=0}^{r-1}\sum_{|\gamma|=j}\bigg(\sum_{\theta\le\alpha}\Qbinom{\alpha}{\theta}(P_j)_{\alpha-\theta,\theta+\beta}\big(\DQ^\gamma f\big)_{\theta,\beta}\mbf{Z}^\gamma\bigg)\\
			&=\wt{P}_{\alpha,\beta}+\sum_{j=0}^{r-1}\sum_{\theta\le\alpha}\Qbinom{\alpha}{\theta}(P_j)_{\alpha-\theta,\theta+\beta}\bigg(\sum_{|\gamma|=j}f_{\theta+\gamma,\beta}\mbf{Z}^\gamma\bigg).
		\end{align*}
		Since \(P_{r-1}(\Q^{\alpha+\beta})=(P_{r-1})_{0^m,\alpha+\beta}\ne0\) for all \(\alpha\in\mb{N}^m,\,|\alpha|\le d+k-1\), we can rewrite the above as
		\begin{align}
			\sum_{|\gamma|=r-1}f_{\alpha+\gamma,\beta}\mbf{Z}^\gamma&=-\frac{1}{(P_{r-1})_{0^m,\alpha+\beta}}\Bigg((\wt{P})_{\alpha,\beta}+\sum_{\substack{(j,\theta)\le(r-1,\alpha)\\(j,\theta)\ne(r-1,\alpha)}}\Qbinom{\alpha}{\theta}(P_j)_{\alpha-\theta,\theta+\beta}\bigg(\sum_{|\gamma|=j}f_{\theta+\gamma,\beta}\mbf{Z}^\gamma\bigg)\Bigg)\label{eq:solve}
		\end{align}
		for all \(\alpha\in\mb{N}^m,\,|\alpha|\le d+k-1\).  So we can set the undetermined coefficients among \(f_{\gamma,\beta}\in\mb{K},\,|\gamma|\le r-2\) freely, and the remaining coefficients of \(f(\mb{X})\) are uniquely determined the above equation.  This means the solution space is a \(\mb{K}\)-affine space having dimension at most \(\binom{m+r-2}{r-2}\).  Further, at any instance of~(\ref{eq:solve}), if we know the R.H.S. and determine the L.H.S., then each coefficient \(f_{\alpha+\gamma,\beta}\) can be recovered by instantiating an arbitrary finite grid \(S^m\) with \(S\subseteq\mb{K},\,|S|\ge ms\) as a hitting set for the \(\mbf{Z}\)-variables, since we have \(\deg_{\mbf{Z}}(P)<ms\).
	\end{proof}

	Our main theorem is now immediate.
	\begin{proof}[Proof of Theorem~\ref{thm:multivariate-LD}]
		Choosing \(r=O(m/\epsilon^m)\), and plugging in all the parameters into~\cref{pro:LR-interpolation} and~\cref{pro:LR-solving} immediately implies that the output list is contained in a \(\mb{K}\)-affine space with dimension at most \(O(m^2/\epsilon)\).
	\end{proof}

		\paragraph*{Acknowledgments.}  The author thanks
		\begin{itemize}[leftmargin=*]
			\item  Dean Doron, Mrinal Kumar, Noga Ron-Zewi, Amnon Ta-Shma, and Mary Wootters for patiently listening to and giving feedback on his undercooked presentations of this work at different times.
			\item  the anonymous reviewers of ITCS 2025 for their critical feedback.
		\end{itemize}
		
	\newpage	
	\raggedright
	\printbibliography

\end{document}